\documentclass[article]{biometrika}

\usepackage{amsmath}

\usepackage{times}
\usepackage{bm}
\usepackage{xr}

\usepackage[utf8]{inputenc}
\usepackage[english]{babel}
\usepackage{amsfonts}
\usepackage{amssymb}
\usepackage{color}
\usepackage{colordvi}
\usepackage{lscape}
\usepackage{rotating}
\usepackage{graphicx}
\usepackage{pgfplots}
\usepackage{tikz}
\usepackage{times}
\usepackage{bm}
\usepackage{natbib}

\usepackage[plain,noend]{algorithm2e}

\pgfplotsset{compat=newest}
\pgfplotsset{plot coordinates/math parser=false}
\newcommand{\tsum}{\textstyle\sum}

\makeatletter
\def\BState{\State\hskip-\ALG@thistlm}
\makeatother
\def\1{1\!{\rm l}}

\newlength\figureheight
\newlength\figurewidth

\makeatletter
\renewcommand{\algocf@captiontext}[2]{#1\algocf@typo. \AlCapFnt{}#2} 
\def\@algocf@capt@plain{top}
\renewcommand{\algocf@makecaption}[2]{%
  \addtolength{\hsize}{\algomargin}%
  \sbox\@tempboxa{\algocf@captiontext{#1}{#2}}%
  \ifdim\wd\@tempboxa >\hsize
    \hskip .5\algomargin%
    \parbox[t]{\hsize}{\algocf@captiontext{#1}{#2}}
  \else%
    \global\@minipagefalse%
    \hbox to\hsize{\box\@tempboxa}
  \fi%
  \addtolength{\hsize}{-\algomargin}%
}
\makeatother

\begin{document}

\jname{}
\jyear{8 May 2018}
\jvol{}
\jnum{}

\received{March 2018}
\revised{May 2018}

\markboth{D. T. Frazier, G. M. Martin, C. P. Robert, \and J. Rousseau}{Asymptotics of Approximate Bayesian Computation}

\title{Asymptotic Properties of Approximate Bayesian Computation}

\author{D. T. FRAZIER, G. M. MARTIN}
\affil{Department of Econometrics and Business Statistics, Monash University, Scenic Boulevard, Clayton Victoria 3800, Australia
\email{david.frazier,gael.martin@monash.edu}}

\author{C. P. ROBERT}
\affil{Universit\'e Paris Dauphine, PSL Research University, 75775 Paris cedex 16, France
\email{xian@ceremade.dauphine.fr}}

\author{J. Rousseau} 
\affil{Statistics Department, University of Oxford, 24-29 St Giles, OX1 3LB, U.K.
\email{judith.rousseau@stats.ox.ac.uk}}
\maketitle

\begin{abstract}
Approximate Bayesian computation allows for statistical analysis in models
with intractable likelihoods. In this paper we consider the asymptotic
behaviour of the posterior distribution obtained by this method. We give
general results on the rate at which the posterior distribution concentrates
on sets containing the true parameter,  its limiting shape, 
 and the asymptotic distribution of the posterior mean.
These results hold under given rates for the tolerance used within the method,
mild regularity conditions on the summary statistics, and a condition linked to
identification of the true parameters. Implications  for
practitioners are discussed.
\end{abstract}


\section{Introduction}

Interest in approximate Bayesian computation methods has begun to shift from
its initial focus as a computational tool toward its validation as a statistical inference procedure;
see, e.g., Fearnhead and Prangle (2012), Marin {et al. }(2014),
Creel and Kristensen (2015), Drovandi {et al}. (2015), Creel et al. (arxiv:1512.07385), {Martin et al. (arxiv:1604.07949)} and Li and Fearnhead (2018a,b). Hereafter we denote these preprints by Creel {et al.} (2015), and Martin {et al.} (2016). 

We study large sample properties of posterior
distributions and posterior means obtained from approximate Bayesian computation algorithms.
Under mild regularity conditions on the underlying summary statistics, we
characterize the rate of posterior concentration and show that the limiting posterior
shape crucially depends on the interplay between the rate
at which the summaries converge and the rate at which the
tolerance used to select parameters shrinks to zero. Bayesian consistency
places a less stringent condition on the speed with which the tolerance
declines to zero than does asymptotic normality of the posterior distribution.
Further, and in contrast to textbook Bernstein--von Mises results, 
asymptotic normality of the  posterior mean does not
require asymptotic normality of the posterior distribution, the former
being attainable under weaker conditions on the tolerance than required for
the latter. Validity of these results requires that the 
summaries converge toward a well-defined limit and
that this limit, viewed as a mapping from parameters to
summaries, be injective. These conditions have a close
correspondence with those required for theoretical validity of indirect
inference and related frequentist estimators, {see, e.g., Gouri\'eroux et al. (1993)}.

We focus on three aspects of asymptotic behaviour: 
posterior consistency, limiting posterior shape, and the
asymptotic distribution of the posterior mean. Our focus is broader than that of
existing studies on the large sample properties of approximate Bayesian
computation algorithms, in which the
asymptotic properties of resulting point estimators have been the
primary focus; see Creel {et al}. (2015) and Li and Fearnhead (2018a). Our approach allows both weaker conditions 
and a complete characterization of the limiting posterior shape. We
distinguish between the conditions, on both the summaries and the
tolerance, required for concentration and those required for 
distributional results. These results suggest how the tolerance in approximate Bayesian computation should be chosen to ensure posterior concentration, valid coverage levels for credible sets, and asymptotically normal and efficient point estimators.

\section{Preliminaries and Background\label{prelim}}

We observe data ${y}=(y_{1},\dots,y_{T})^\intercal$, $T\geq 1$, drawn
from the model $\{P_{{\theta }}:{\theta \in \Theta }\}$, where 
$P_{{\theta }}$ admits the corresponding conditional density $p(\cdot
\mid{\theta })$, and ${\theta }\in {\Theta }\subset \mathbb{R}^{k_{\theta }}$.
Given a prior measure $\Pi(\theta)$ with density $\pi({\theta })$, the aim of the algorithms under study is to
produce draws from an approximation to the exact posterior density 
$
\pi({\theta \mid y})\propto p({y\mid\theta })\pi({\theta })
$, when both parameters and pseudo-data $({\theta },{z})$ can
easily be simulated from $\pi({\theta })p({z\mid\theta })$, but  $p({z\mid\theta })$ is intractable. The simplest accept/reject form of the
algorithm (Tavar\'{e}{ {et al.}, }1997; Pritchard {{et al.}},
1999) is detailed in Algorithm 1. 

\bigskip
\noindent Algorithm 1.
Approximate Bayesian Computation 
\begin{tabbing}
   \enspace (1) Simulate ${\theta }^{i}$ $(i=1,\dots,N)$ from $\pi({\theta }),$\\
   \enspace (2) Simulate ${z}^{i}=(z_{1}^{i},\dots,z_{T}^{i})^{\intercal}$ $(i=1,\dots,N)$ from the likelihood, $p(\cdot\mid{\theta }^{i})$,\\
\enspace (3) Select ${\theta }^{i}$ such that
$
d\{{\eta }({y}),{\eta }({z}^{i})\}\leq
\varepsilon , 
$ 
where ${\eta }({\cdot })$ is a statistic, $d(\cdot
, \cdot )$ is a distance function, \\ and $\varepsilon>0$ is the tolerance level.
\end{tabbing}

Algorithm 1 thus samples ${\theta }$ and ${z}$ from the joint posterior density
$$
\pi_{\varepsilon }\{{\theta },{z\mid\eta }({y})\}={\pi({\theta })p({z\mid\theta }%
)1\!\mathrm{l}_{\varepsilon }({z})}\big/\int\pi({%
\theta })p({z\mid\theta })1\!\mathrm{l}_{\varepsilon }({z})d{z}d{\theta},
$$
where $1\!\mathrm{l}_{\varepsilon }({z})=1\!\mathrm{l}[d\{{\eta }({y}),{\eta 
}({z})\}\leq \varepsilon ]=1$ if $d\left\{ {\eta }({y}),{\eta }({z}%
)\right\} \leq \varepsilon $, and zero otherwise. The approximate Bayesian computation posterior density is defined as
$$
\pi_{\varepsilon }\{{\theta \mid\eta }({y})\}=\textstyle\int_{{}}\pi_{\varepsilon }\{{%
\theta },{z\mid\eta }({y})\}d{z}.  
$$
Below, we refer to $\pi_{\varepsilon }\{{\theta \mid\eta }({y})\}$ as the approximate posterior density. Likewise, the posterior probability of a set $A\subset\Theta$ associated with Algorithm 1 is 
\begin{equation*}
\Pi _{\varepsilon }\{A\mid \eta (y)\}=\Pi \left[ A\mid d_{}\{{\eta }({y}),{\eta }({z}%
)\}\leq \varepsilon \right] =\int_{A}\pi_{\varepsilon }\{{\theta }\mid {\eta }({y}%
)\}d{\theta },
\end{equation*}and we refer to $\Pi_{\varepsilon}\{\cdot\mid \eta(y)\}$ as the approximate posterior distribution. When $\eta(\cdot)$ is sufficient for the observed data $y$ and $\varepsilon$ is close to zero, $\pi_{\varepsilon }\{{\theta \mid\eta }({y})\}$ will be a good approximation to $\pi({\theta \mid y})$, and draws of ${\theta }$
from $\pi_{\varepsilon }\{\theta \mid\eta({y})\}$ can be used to estimate
features of $\pi(\theta \mid y)$. 

In practice $\eta(y)$ is rarely sufficient for $y$, and draws of $\theta$
can only be used to approximate $\pi\{\theta \mid \eta(y)\}=\lim_{\varepsilon\rightarrow0}\pi_{\varepsilon}\{\theta \mid \eta(y)\}$. Given the
general lack of sufficient statistics, we need to assess the behavior of the
approximate posterior distribution $\Pi_{\varepsilon }\{\cdot\mid \eta(y)\}$,
and to establish whether or not $\Pi_{\varepsilon }\{\cdot\mid \eta(y)\}$
behaves in a manner that is appropriate for statistical inference, with asymptotic theory being one obvious approach. 

Establishing the large sample behavior of $\Pi_{\varepsilon }\{\cdot\mid \eta(y)\}$,
including point and interval estimates derived from this distribution,
gives practitioners 
guarantees on the reliability of approximate Bayesian computations.
Furthermore, these results allow us to provide guidelines for choosing the tolerance $\varepsilon$ so that $\Pi_{\varepsilon }\{\cdot\mid \eta(y)\}$ possesses desirable statistical properties.

Before presenting our results, we set notation used
throughout the paper. Let $\mathcal{B}\subset \mathbb{R}^{k_{\eta }}$ denote
the range of the simulated summaries $\eta(z)$. Let $d_{1}(\cdot ,\cdot )$ be a
metric on ${\Theta }$ and $d_{2}(\cdot ,\cdot)$ a metric on $\mathcal{B}$. Take $\Vert \cdot \Vert $ to be the Euclidean norm. Throughout, $C$
denotes a generic positive constant. For real-valued sequences $\{a_{T}\}_{T\geq 1}$ and
$\{b_{T}\}_{T\geq 1}$, $a_{T}\lesssim b_{T}$ denotes $a_{T}\leq Cb_{T}$ for
some finite $C>0$ and $T$ large, $a_{T}\asymp b_{T}$ implies  that $a_{T}\lesssim b_{T} \lesssim a_{T}$, and $a_{T}{\gg }b_{T}$ indicates a larger order of
magnitude. For $x_{T}$ a random variable, $x_{T}=o_{P}(a_{T})$ if
$\lim_{T\rightarrow \infty }\text{pr} (|x_{T}/a_{T}|\geq C)=0$ for any $C>0$
and $x_{T}=O_{P}(a_{T})$ if for any $C>0$ there exists a finite $M>0$ and a
finite $T$ such that $\text{pr}(|x_{t}/a_{t}|\geq M)\leq C$, for all $t>T$. All limits are taken as $T\rightarrow\infty$. When no confusion will result, $\lim_{T}$ replaces $\lim_{T\rightarrow\infty}$.

\section{Concentration of the Approximate Bayesian Computation Posterior 
\label{Aux}}

We assume throughout that the model is correctly specified: for some $\theta_0$ in the interior of $\Theta$, we have $P_{{\theta }}=P_{0}$.  Asymptotic validity of any Bayesian procedure
requires posterior concentration, which is often referred to as Bayesian consistency. In our 
context, this equates to the following posterior
concentration property: for any $\delta >0$, and for some $\varepsilon>0$, 
$$
\Pi _{\varepsilon }\{d_{1}(\theta ,\theta _{0})>\delta\mid \eta (y)\}=\Pi \left[ d_{1}(\theta ,\theta _{0})>\delta \mid d_{2}\{{\eta }({y}),{\eta }({%
z})\}\leq \varepsilon \right] =\int_{d_{1}(\theta ,\theta _{0})>\delta
}\pi_{\varepsilon }\{\theta \mid \eta (y)\}d\theta =o_{P}(1).   
$$
This property is paramount since, for any $A\subset \Theta $, $\Pi \left[
A\mid d_{2}\{{\eta }({y}),{\eta }({z})\}\leq \varepsilon \right] $ will differ from
the exact posterior\ probability.
Without the guarantees of exact posterior inference,
knowing that $\Pi_{\varepsilon}\{\cdot\mid\eta(y)\}$ will 
concentrate on $\theta _{0}$ gives validity to its use as a means of expressing our uncertainty about $\theta$.

Posterior concentration is related to the rate at which information about $%
\theta _{0}$ accumulates in the sample. The amount of information Algorithm %
1 provides depends on the rate
at which the observed summaries $\eta(y)$ and the simulated summaries $\eta(z)$ converge to well-defined
limit counterparts ${b}({\theta }_{0})$ and ${b}({\theta })$, and the
rate at which information about $\theta _{0}$ accumulates within the
algorithm, governed by the rate at which $\varepsilon $ goes to $0$%
. To link both factors we consider $\varepsilon $ as a $T$-dependent
sequence $\varepsilon _{T}\rightarrow 0$ as ${T\rightarrow \infty }$. 
We can now state the technical assumptions used
to establish our first result. These assumptions are
applicable to a broad range of data structures, including 
weakly dependent data.

\begin{assumption}\label{[A1]}
There exist a non-random map ${b}:{\Theta }%
\rightarrow \mathcal{B}$, and a sequence of functions $\rho _{T}(u)$ that are monotone non-increasing in $u$ for any $T$ and satisfy $\rho _{T}(u)\rightarrow 0$ as $%
T\rightarrow \infty $. For fixed $u$, and for all $\theta\in\Theta$,
\begin{equation*}
P_{{\theta }}\left[ d_{2}\{{\eta }({z}),{b}({\theta })\}>u\right] \leq c({%
\theta })\rho _{T}(u),\quad \int_{\Theta }c({\theta })d\Pi ({\theta }%
)<\infty,
\end{equation*}%
with either of the following assumptions on $c(\cdot )$:

\noindent{(i)} there exist $c_{0}<\infty $ and $\delta >0$ such that
for all ${\theta }$ satisfying $d_{2}\{{b}({\theta }%
),{b}({\theta}_{0})\}\leq \delta $ then $c({\theta })\leq c_{0}$;

\noindent{(ii)} there exists $a>0$ such that $\int_{\Theta }c({\theta }%
)^{1+a}d\Pi ({\theta })<\infty .$
\end{assumption}

\begin{assumption} \label{[A2]} There exists some $D>0$ such that, for all $\xi >0$
and some $C>0$, the prior probability satisfies
$\Pi \left[ d_{2}\{{b}({\theta }),{b}({\theta }_{0})\}\leq \xi \right]
\geq C \xi ^{D}.$
\end{assumption}

\begin{assumption}\label{[A3]} {(i)} The map ${b}$
is continuous. {(ii)} The map ${b}$
is injective and satisfies:
$\Vert {\theta }-{\theta }_{0}\Vert \leq L\Vert {b}({\theta })-{b}({\theta }%
_{0})\Vert ^{\alpha }$ on some open neighbourhood of ${\theta }_{0}$ with $L>0$ and $\alpha >0$.
\end{assumption}

\begin{remark}\label{rem2}
The convergence of $\eta(z)$ to $b(\theta)$ in Assumption \ref{[A1]} is the key to
posterior concentration and without it, or a similar assumption,
Bayesian consistency will not occur. The function $\rho_{T}(u)$ in Assumption \ref{[A1]} typically takes the form $\rho_{T}(u)=\rho_{}(uv_{T})$, for $v_{T}$ a sequence such that $d_{2}\{\eta(z),b(\theta)\}=O_{P}(1/v_{T})$, and where $\rho(u v_{T})$ controls the tail behavior of $d_{2}\{\eta(z),b(\theta)\}$. The specific structure of $\rho(uv_{T})$ will depend on what is assumed about the properties of the underlying summaries $\eta(z)$. In most cases, $\rho(uv_{T})$ will have either a polynomial or exponential structure in $uv_{T}$, and thus satisfy one of the following rates.
\smallskip

\noindent{(a)} {Polynomial:} there exist a diverging positive sequence $\{v_{T}\}_{T\ge 1}$
and $u_{0},\kappa>0$ such that
\begin{equation}
P_{\theta}\left[d_{2}\{\eta(z),b(\theta)\}>u\right]\leq c(\theta)\rho_{T}(u),\quad\rho _{T}(u)=1\big/(uv_{T})^{\kappa },\quad u\leq u_{0},
\label{rho_poly}
\end{equation}where, for some $c_0>0$ and $\delta>0$, $\int_{\Theta}c(\theta)d\Pi(\theta)<\infty$ and if $d_{2}\{b(\theta),b(\theta_0)\}\leq\delta$, then $c(\theta)\leq c_0$.
\smallskip 

\noindent{(b)} {Exponential:} there exist $h_{\theta}(\cdot)>0$ and $u_0>0$ such that
\begin{equation}
P_{\theta}\left[d_{2}\{\eta(z),b(\theta)\}>u\right]\leq c(\theta)\rho_{T}(u),\quad\rho _{T}(u)=\exp\{-h_{{\theta }}(uv_{T})\},\quad u\leq u_{0}, \label{rho_exp}
\end{equation}
where, for some $c,C>0$, $\int_{\Theta} c(\theta)\exp\{-h_{\theta}(u v_{T})\}d\Pi(\theta)\leq C\exp\{-c(u v_{T})^{\tau}\}.$ \smallskip

\noindent To illustrate these cases for $\rho_{T}(\cdot)$, consider the summary statistics $\eta ({z})=T^{-1}\sum_{i=1}^{T}g(z_{i})$ where, for simplicity, $\{g(z_{i})\}_{i\leq T}$ is independent and identically distributed, and $b({\theta })=E_{{\theta }}\{g(Z)\}$. 

If $g(z_{i})-b(\theta)$ has a
finite moment of order $\kappa $, $\rho_{T}(u)$ will satisfy \eqref{rho_poly}: from Markov's
inequality, 
\begin{equation*}
P_{\theta}\left\{\Vert {\eta }({z})-b({\theta })\Vert >u\right\}\leq {%
C_{}E_{{\theta }}\left\{ | g(Z)|^{\kappa }\right\} }\big/{({uT^{1/2}})^{\kappa }%
}.
\end{equation*}%
With reference to \eqref{rho_poly}, $\rho_{T}(u)=1/(uv_{T})^{\kappa}$, $v_{T}=T^{1/2}$ and $c(\theta)=C\mathbb{E}_{\theta}\{|g(Z)|^{\kappa}\}<\infty$. If the map ${\theta }\mapsto E_{{\theta }}\left\{
|g(Z)|^{\kappa }\right\}$ is continuous at ${\theta }_{0}$ and positive, 
Assumption \ref{[A1]} is satisfied. 

If $\{g(z_i)-b(\theta)\}$ has a finite exponential moment, $\rho_{T}(u)$ will satisfy \eqref{rho_exp}: from a version of the Bernstein inequality,
\begin{equation*}
\begin{split}
P_{\theta}\left\{ \Vert {\eta }({z})-b({\theta })\Vert >u\right\}
&\leq \exp\left[-u^{2}T/\{2c(\theta )\}\right].
\end{split}%
\end{equation*}%
With reference to (\ref{rho_exp}), $\rho_{T}(u)=\exp\{-h_{\theta}(uv_{T})\}$, $h_{{\theta }}(uv_{T})=u^{2}v_{T}^{2}/\{2c(\theta )\}$ and $v_{T}=T^{1/2}$.
If the map ${\theta }\mapsto c(\theta )$ is continuous at ${\theta }_{0}$ and positive, Assumption 
\ref{[A1]} is satisfied.
\end{remark}
\begin{remark}\label{remark2}  Assumption \ref{[A2]} controls the
degree of prior mass in a neighbourhood of ${\theta }_{0}$ and is standard
in Bayesian asymptotics. For $\xi$ small, the larger $D$,
the smaller the amount of prior mass near $\theta _{0}$. If the prior measure $\Pi(\theta)$ is absolutely continuous with prior density $\pi(\theta)$ and
if $\pi$ is bounded, above and below, near $\theta _{0}$, then $%
D=\dim (\theta )=k_{\theta }$. Assumption \ref{[A3]} is an identification condition
that is critical for obtaining posterior concentration around ${%
\theta }_{0}$. Injectivity of $b$ depends on both the true structural model and the particular
choice of ${\eta }$. Without this identification condition posterior concentration 
at $\theta_0$ 
cannot occur. 
\end{remark}

\begin{theorem}\label{thm1} If Assumptions \ref{[A1]}--\ref{[A2]} are satisfied, then, for $M$ large enough, as $T\rightarrow\infty$ and $\varepsilon_T=o(1)$, with $P_{0}$ probability going to one,
\begin{equation}
\Pi \left[ d_{2}\{{b}({\theta }),{b}({\theta }_{0})\}>\lambda_{T}\mid d_{2}\{{\eta }({y}),{\eta }({z}%
)\}\leq \varepsilon _{T}\right] \lesssim 1/M,  \label{i}
\end{equation} with $\lambda_{T}=4\varepsilon_{T}/3 + \rho_{T}^{-1}(\varepsilon_{T}^{D}/M).$ Moreover, if Assumption \ref{[A3]} also holds, as $T\rightarrow\infty$,
\begin{equation}
\Pi \left[ d_{1}({\theta },{\theta }_{0})>L\lambda_{T}^{\alpha }\mid d_{2}\{{\eta }({y}),{\eta }({z}%
)\}\leq \varepsilon _{T}\right] \lesssim 1/M,  \label{ii}
\end{equation}
\end{theorem}

Since equations \eqref{i} and \eqref{ii} hold for any $M$ large enough, we can conclude that the posterior distribution behaves like an $o_{P}(1)$ random variable on sets that do not include $\theta_0$, and Bayesian consistency of $\Pi_{\varepsilon}\{\cdot \mid \eta(y)\}$ follows. More generally, \eqref{i} and \eqref{ii} give a posterior concentration rate, denoted by $\lambda_{T}$ in Theorem \ref{thm1}, that depends on $\varepsilon_{T}$ and on the underlying behavior of $\eta(z)$, as described by $\rho_{T}(u)$. We must consider the nature of this concentration rate in order to understand which choices for $\varepsilon_{T}$ are appropriate under different assumptions on the summary statistics.

As mentioned above, the deviation control function $\rho_{T}(u)$ will often be of a polynomial \eqref{rho_poly} or exponential \eqref{rho_exp} form. Under these two assumptions, $\rho_{T}(u)$ has an explicit representation and the concentration rate $\lambda_{T}$ can be obtained by solving the equation $$\lambda_{T}= 4\varepsilon_{T}/3 + \rho_{T}^{-1}(\varepsilon_{T}^{D}/M).$$

\noindent (a) {polynomial case:} From equation (\ref{rho_poly}), the deviation control function is
$\rho_{T}(u)=1/(uv_{T})^{\kappa}$. To obtain the posterior concentration rate, we invert $\rho_{T}(u)$ to obtain $\rho _{T}^{-1}(\varepsilon
_{T}^{D})=1/(\varepsilon _{T}^{D/\kappa }v_{T})$, and then equate $%
\varepsilon _{T}$ and $\rho _{T}^{-1}(\varepsilon _{T}^{D})$, to obtain $%
\varepsilon _{T}\asymp v_{T}^{-\kappa /(\kappa +D)}$. This choice of $\varepsilon_{T}$ implies concentration of the approximate posterior distribution at the  rate 
\begin{equation*}
\lambda _{T}\asymp v_{T}^{-\kappa /(\kappa +D)}.
\end{equation*}

\smallskip
\noindent (b) {exponential case:} If the summary statistics admit an exponential moment, a faster rate of posterior concentration obtains. From equation \eqref{rho_exp}, $\rho _{T}(u)=\exp\{-h_{{\theta }}(uv_{T})\}$
and there exist finite $u_{0},c,C>0$ such that 
\begin{equation*}
\int_{\Theta }c({\theta })e^{-h_{\theta }(uv_{T})}d\Pi ({\theta })\leq
Ce^{-c(uv_{T})^{\tau }},\quad u\leq u_{0}.
\end{equation*}%
Hence if $c({\theta })$ is bounded from above and if $h_{{\theta }}(u)\ge
u^{\tau }$ for ${\theta }$ in a neighbourhood of $\theta_0$, then $\rho
_{T}(u)\asymp \exp\{-c_{0}(uv_{T})^{\tau }\}$; thus, $\rho
_{T}^{-1}(\varepsilon _{T}^{D})\asymp (-\log \varepsilon _{T})^{1/\tau
}/v_{T}$. Following arguments   similar  to those used in (a) immediately above, if we take $\varepsilon _{T}\asymp (\log v_{T})^{1/\tau }/v_{T}$,
the approximate posterior distribution concentrates at the rate
\begin{equation*}
\lambda _{T}\asymp {(\log v_{T})^{1/\tau }}/{v_{T}}. 
\end{equation*}%

\begin{example} We now illustrate the conditions of Theorem %
\ref{thm1} in a 
moving average model of order two:%
$$
y_{t}=e_{t}+\theta _{1}e_{t-1}+\theta _{2}e_{t-2}\;\;(t=1,\dots,T),  
$$
where $\{e_{t}\}_{t=1}^{T}$ is a sequence of white noise random variables such
that $E(e_{t}^{4+\delta })<\infty $ and some $\delta >0$. Our prior 
for $\theta=(\theta _{1},\theta _{2})^{\intercal}$ is uniform over the
following invertibility region, 
\begin{equation}
-2\leq \theta _{1}\leq 2,\;\theta _{1}+\theta _{2}\geq -1,\theta _{1}-\theta
_{2}\leq 1.  \label{const1}
\end{equation}%
Following Marin {et al.} (2011), we choose as summary statistics for
Algorithm 1 the sample autocovariances $\eta _{j}({y})=T^{-1}\sum_{t=1+j}^{T}y_{t}y_{t-j}$, for $j=0,1,2$.
For this choice the $j$-th component
of $b(\theta )$ is $b_{j}(\theta )=E_{\theta }(z_{t}z_{t-j})$.

Now, take $d_{2}\{\eta (z),b(\theta )\}=\left\Vert \eta (z)-b(\theta
)\right\Vert $. Under the moment condition for $e_{t}$ above, it follows that $V(\theta )=E[\{\eta (z)-b(\theta )\}\{\eta (z)-b(\theta
)\}^{\intercal }]$ satisfies $\text{tr}\{V(\theta )\}<\infty $ for all $\theta 
$ in \eqref{const1}. By an application of Markov's inequality, we can
conclude that 
\begin{flalign*}
P_{\theta}\left\{\|\eta(z)-b(\theta)\|>u\right\}=P_{\theta}\left\{\|\eta(z)-b(\theta)\|^{2}>u^2\right\}&\leq \frac{\text{tr}\{V(\theta)\}}{u^2 T}+o(1/T),
\end{flalign*}where the $o(1/T)$ term comes from the fact that there are
finitely many non-zero covariance terms due to the $m$-dependence of
the series, and Assumption \ref{[A1]} is satisfied. 
Given the structure of $b(\theta )$, the uniform prior $\pi(\theta )$
over \eqref{const1} fulfills Assumption \ref{[A2]}. Furthermore, $\theta\mapsto b(\theta)=(1+\theta _{1}^{2}+\theta _{2}^{2},(1+\theta _{2})\theta
_{1},\theta _{2})^{\intercal }$ is  injective  and satisfies Assumption \ref{[A3]}. As noted in Remark \ref{remark2}, the injectivity of $\theta\mapsto b(\theta)$ is required for posterior concentration, and without it there is no guarantee that the posterior will concentrate on $\theta_0$. Since the sufficient conditions for Theorem \ref{thm1} are satisfied, approximate Bayesian computation based on this choice of statistics will yield an approximate posterior density that concentrates on $\theta_0$.

\end{example}

Theorem \ref{thm1} can also be
visualized by fixing a particular value of $\theta $, say $%
\tilde{\theta}$, and generating observed data sets $\tilde{y}$
of increasing length, then running Algorithm 1 on these 
data sets. If the conditions of Theorem \ref{thm1} are
satisfied, the approximate posterior density will become increasingly peaked at $\tilde{\theta}$ as $T$ increases. Using  Example 1, we demonstrate this behavior in the
Supplementary Material.

\section{Shape of the Asymptotic Posterior Distribution\label{norm}}
\subsection{Assumptions and Theorem} 
While posterior concentration states that $\Pi \left[ d_{1}(\theta
,\theta _{0})>\delta \mid d_{2}\{{\eta }({y}),{\eta }({z})\}\leq \varepsilon
_{T}\right] =o_{P}(1)$ for an appropriate choice of $\varepsilon _{T}$, it
does not indicate precisely how this mass accumulates, or the approximate
amount of posterior probability within any neighbourhood of $\theta _{0}$. This
information is needed to obtain accurate expressions of
uncertainty about point estimators of $\theta _{0}$ and to ensure that credible
regions have proper frequentist coverage. To this end, we now
analyse the limiting shape of 
$\Pi \left[ \cdot \mid d_{2}\{{\eta }({y}),{\eta }({z})\}\leq \varepsilon _{T}%
\right] $ for various relationships between $\varepsilon _{T}$ and the rate
at which summary statistics satisfy a central limit theorem. In this and the following sections, we denote $\Pi \left[ \cdot\mid d_{2}\{{\eta }({y}),{\eta }({z})\}\leq \varepsilon _{T}\right] $ by $\Pi
_{\varepsilon }\{\cdot \mid {\eta }(y)\}$. Let $%
\Vert \cdot \Vert _{\ast }$ denote the spectral norm.

In addition to Assumption \ref{[A2]} the following conditions
are needed to establish the results of this section.

\begin{assumption}\label{[A1']}
 Assumption \ref{[A1]} holds. There exists a sequence of positive definite matrices $\{\Sigma _{T}({\theta }%
_{0})\}_{T\ge 1}$, $c_0>0$, $\kappa >1$ and $\delta >0$ such that for all $\Vert {\theta }-{%
\theta }_{0}\Vert \leq \delta $, $P_{{\theta }}\left[ \Vert \Sigma _{T}({%
\theta }_{0})\{{\eta }({z})-{b}({\theta })\}\Vert >u\right] \leq {c_{0}%
}{u^{-\kappa }}$ for all $0<u\leq \delta \|\Sigma_T(\theta_0)\|_*$, uniformly in $T$.
\end{assumption}
 
\begin{assumption}\label{[A3']} Assumption \ref{[A3]} holds. The
map $\theta\mapsto{b}(\theta)$ is continuously differentiable at ${\theta
_{0}}$ and the Jacobian $\nabla _{\theta }b({\theta }_{0})$ has full column
rank $k_{\theta }$.
\end{assumption}

\begin{assumption}
\label{[A4]} The value
$\theta_0$ is in the interior of $\Theta$. For some $\delta >0$  and for all $\Vert {%
\theta }-{\theta }_{0}\Vert \leq \delta $, there exists a sequence of $(k_{\eta }\times k_{\eta
}) $ positive definite matrices $\{\Sigma_{T}({\theta })\}_{T\ge 1}$, with $k_{\eta
}=\dim \{\eta (z)\}$, such that for all open sets $B$
\begin{equation*}
\sup_{|\theta - \theta_0|\leq \delta} \left| P_\theta\left[ {\Sigma }_{T}({\theta })\{{\eta }({z})-{b}({\theta })\} \in B\right] - P\left\{ \mathcal{N%
}(0,I_{k_{\eta }}) \in B\right\} \right| \rightarrow 0
\end{equation*} in distribution as $T\rightarrow\infty$, where $I_{k_{\eta }}$ is the $(k_{\eta }\times k_{\eta })$ identity
matrix.
\end{assumption}
\begin{assumption}
\label{[A5]} There exists $v_T\rightarrow\infty$ such that for all $\Vert \theta -\theta _{0}\Vert \leq \delta $%
, the sequence of functions ${\theta }\mapsto {\Sigma }_{T}({\theta }%
)v_{T}^{-1}$ converges to some positive definite $A(\theta )$ and is
equicontinuous at ${\theta }_{0}$. \smallskip
\end{assumption}

\begin{assumption}\label{[A6]} For some positive $\delta $, all $\Vert {\theta }-{\theta }_{0}\Vert \leq \delta $, all ellipsoids $B_{T}=\big\{%
(t_{1},\dots ,t_{k_{\eta }}):\sum_{j=1}^{k_{\eta }}t_{j}^{2}/h_{T}^{2}\leq 1%
\big\}$ and all $u\in \mathbb{R}^{k_{\eta }}$ fixed, for all $%
h_{T}\rightarrow 0$, as $T\rightarrow \infty $, 
\begin{equation*}
\begin{split}
\lim_{T}
\sup_{|\theta - \theta_0|\leq \delta } \left| h_{T}^{-k_{\eta }} P_{\theta }\left[ {\Sigma }_{T}({\theta })\{{\eta }({z})-{b}(%
{\theta })\}-u\in B_{T}\right]  - 
 \varphi _{k_{\eta}}(u) \right| &= 0
, \\
{h_{T}^{-k_{\eta }}}\,{P_{{\theta }}\left[ {\Sigma }_{T}({\theta })\{{\eta }({z})-{b}(%
{\theta })\}-u\in B_{T}\right] } & \leq H(u),\quad \int
H(u)du<\infty ,
\end{split}
\end{equation*}%
for $\varphi _{k_{\eta }}(\cdot )$ the density of a $k_{\eta }$-dimensional
normal random variate.
\end{assumption}
\smallskip

\begin{remark}
 Assumption \ref{[A1']} is similar to Assumption \ref{[A1]} but for $\Sigma _{T}(\theta _{0})\{\eta
(z)-b(\theta )\}$. Assumption \ref{[A4]} is a
central limit theorem for $%
\{\eta (z)-b(\theta )\}$ and, as such, requires the existence of a 
positive-definite matrix $\Sigma _{T}(\theta )$. In simple cases,
such as independent and identically distributed data with $\eta(z)=T^{-1}\sum_{i=1}^{T}g(z_{i})$, $\Sigma
_{T}(\theta )=v_{T}A_{T}^{}(\theta )$ with $A_{T}(\theta)=A(%
\theta)+o_{P}(1)$ and $V(\theta )=E[\{g(Z)-b(\theta )\}\{g(Z)-b(\theta )\}^{\intercal }]=\{A(\theta)^{\intercal}A(\theta)\}^{-1}$. Assumptions \ref{[A3']}  and \ref{[A6]}  ensure that $\theta \mapsto b(\theta )$ and the covariance
matrix of $\{\eta (z)-b(\theta )\}$ are
well-behaved, which allows the posterior behavior of a normalized version of $%
(\theta -\theta _{0})$ to be governed by that of $\Sigma
_{T}(\theta _{0})\{b(\theta )-b(\theta _{0})\}$. Assumption \ref{[A6]}
governs the pointwise convergence of a normalized version of the measure $%
P_{\theta }$, therein dominated by $H(u)$, and allows the application of the dominated
convergence theorem in Case {(iii)} of the following result.
\end{remark}
\begin{theorem}
\label{normal_thm} Under Assumptions \ref{[A2]}, \ref{[A1']}--\ref{[A5]}, with $\kappa
>k_{\theta }$, the following hold with probability going to 1.

\noindent{(i)} If $\lim_{T}v_{T}\varepsilon _{T}=\infty $, the posterior distribution of $\varepsilon
_{T}^{-1}({\theta }-{\theta }_{0})$ converges to the uniform distribution
over the ellipsoid $\{w:w^{\intercal }B_{0}w\leq 1\}$ with $B_{0}=\nabla _{\theta
}b({\theta _{0}})^{\intercal }\nabla _{\theta }b({\theta _{0}})$, meaning that
for $f(\cdot)$ continuous and bounded,
$$
\int f\{\varepsilon _{T}^{-1}({\theta }-{\theta }_{0})\}d\Pi _{\varepsilon }\{%
{\theta }\mid {\eta(y) }\}\overset{}{ \rightarrow }{\int_{u^{\intercal
}B_{0}u\leq 1}f(u)du}\Big/{\int_{u^{\intercal }B_{0}u\leq 1}du},  
\quad T\rightarrow\infty. 
$$

\noindent{(ii)} If $\lim_{T}v_{T}\varepsilon _{T}=c>0$, there exists a
non-Gaussian distribution on $\mathbb{R}^{k_{\eta }}$, $Q_{c}$, such that
$$
\Pi _{\varepsilon }\left[ {\Sigma }_{T}({\theta }_{0})\nabla_{\theta}b(\theta_0)({\theta }-
\theta_{0})-{\Sigma }_{T}({\theta }_{0})\{\eta(y)-b(\theta_0)\}\in B\mid {\eta }(y)\right] \rightarrow Q_{c}(B),
\quad T\rightarrow\infty.
$$
In particular, $Q_{c}(B) \propto \int_{B}\int_{\mathbb{R}^{k_{\eta}} }\1\!\{(z-x)^{\intercal}A(\theta_0)^{\intercal}A(\theta_0)(z-x)\leq c\}{\varphi}_{k_\eta}(z) dzdx.$

\noindent{(iii)} If $\lim_{T}v_{T}\varepsilon _{T}=0$ and Assumption \ref{[A6]} holds  then, 
$$
\Pi _{\varepsilon }\left[ {\Sigma }_{T}({\theta }_{0})\nabla_{\theta}b(\theta_0)({\theta }-
\theta_{0})-\Sigma_{T}({\theta }_{0})\{\eta(y)-b(\theta_0)\}\in B\mid {\eta }(y)\right] \rightarrow\int_{B}\varphi_{k_{\eta}}(x)dx,\quad  T\rightarrow\infty.
$$
\end{theorem}

\begin{remark}
 Theorem \ref{normal_thm} generalizes to
the case where the components of $\eta (z)$ have different rates of
convergence. The statement and proof of this more general result are
deferred to the Supplementary Material. Furthermore, as with Theorem \ref{thm1}, the behavior of $\Pi_{\varepsilon}\{\cdot\mid \eta(y)\}$ described by Theorem \ref{normal_thm} can be visualized.
This is demonstrated in  the Supplementary Material. Formal verification of the conditions underpinning Theorem 2 is quite
challenging, even in this case. Numerical results nevertheless suggest that for this particular
choice of model and summaries a Bernstein--von Mises result holds, conditional on $\varepsilon_{T}=o(1/v_{T})$, with $v_{T}={T}^{1/2}$.
\end{remark}
\subsection{Discussion of the Result}

Theorem \ref{normal_thm} asserts that
the crucial feature in determining the limiting shape of
$\Pi_{\varepsilon}\{\cdot\mid \eta(y)\}$ is the behaviour of $v_{T}\varepsilon _{T}$. The implication of Theorem \ref%
{normal_thm} is that only in the regime where $\lim_{T} v_{T}\varepsilon_{T}=0$ will $100(1-\alpha)\%$ Bayesian credible regions calculated from $\Pi_{\varepsilon}\{\cdot\mid \eta(y)\}$ have frequentist coverage of $100(1-\alpha)\%$. If 
$\lim_{T}v_{T}\varepsilon _{T}=c>0$, for $c$ finite, $\Pi_{\varepsilon}\{\cdot\mid \eta(y)\}$ is not asymptotically Gaussian and credible regions will have incorrect magnitude, i.e., the coverage will not be at the nominal level. If $\lim_T v_{T}\varepsilon_{T}=\infty$, i.e., $\varepsilon _{T}\gg v_{T}^{-1}$, credible regions will have coverage that converges to 100\%.

In Case 
{(i)}, which corresponds to a large tolerance $\varepsilon _{T}$,
the approximate posterior distribution has nonstandard asymptotic behaviour. In Case (i) $\Pi_{\varepsilon}\{\cdot\mid \eta(y)\}$ behaves like the prior distribution over $\{\theta:\Vert \nabla _{\theta }b({\theta }_{0})({\theta }-{\theta }_{0})\Vert \leq \varepsilon _{T}\{1+o_{P}(1)\}\}$, which, by prior continuity, implies that $\Pi_{\varepsilon}\{\cdot\mid \eta(y)\}$ is equivalent to a
uniform distribution over this set. Li and Fearnhead (2018a) also establish this behaviour, and observe that asymptotically $\Pi_{\varepsilon}\{\cdot\mid \eta(y)\}$ behaves like a convolution of a Gaussian distribution, with variance of order $1/v_T^2$, and a uniform distribution over a ball of radius $\varepsilon_T$, and, where, depending on the order $v_T\varepsilon_T $, one distribution will dominate.

Assumption \ref{[A6]} applies to random
variables ${\eta }({z})$ that are absolutely continuous with respect to the
Lebesgue measure, or in the case of sums of random variables, to sums 
{that} are non-lattice; see Bhattacharya and Rao (1986). For
discrete ${\eta }({z})$, Assumption \ref{[A6]} must be adapted for
Theorem \ref{normal_thm} to be satisfied. One such adaptation is 
\smallskip

\begin{assumption}
There exist $\delta >0$ and a
countable set $E_{T}$ such that for all $\left\Vert {\theta }-{\theta }%
_{0}\right\Vert <\delta $, for all $x\in E_{T}$ such that $\text{ pr}\left\{ {\eta }({z})={x}\right\} >0$, $\text{ pr}\left\{ {\eta }({z})\in E_{T}\right\} =1$
and 
\begin{equation*}
\sup_{\left\Vert {\theta }-{\theta }_{0}\right\Vert \leq \delta }\sum_{x\in
E_{T}}\left\vert p[{\Sigma }_{T}({\theta })\{{x}-{b}({\theta })\}\mid {\theta }%
]-v_{T}^{-k_{\eta}} | A(\theta_0) |^{-1/2}\varphi_{k_{\eta}} \lbrack {\Sigma }_{T}({\theta })\{{x%
}-{b}({\theta })\}]\right\vert =o(1).
\end{equation*}
\end{assumption}
This is satisfied when ${\eta }({z})$ is a sum of independent lattice random variables,
as in the population genetics experiment detailed in Section 3.3 of Marin {{et al.}} (2014), which compares evolution scenarios of separated
populations from a most recent common ancestor. Furthermore, this example satisfies  Assumptions \ref{[A2]} and \ref{[A1']}--\ref{[A5]}. Thus the
conclusions of both Theorems \ref{thm1} and \ref{normal_thm} apply to this
model.

\section{Asymptotic Distribution of the Posterior Mean\label{mean}}

\subsection{Main Result}

The literature on the asymptotics of approximate
Bayesian computation has so far focused primarily on
asymptotic normality of the posterior mean. The posterior normality result in
Theorem \ref{normal_thm} is not weaker, or stronger, than the
asymptotic normality of an approximate point estimator, as the results
focus on different objects. However, existing proofs for asymptotic
normality of the posterior mean all require asymptotic normality of the
posterior distribution $\Pi_{\varepsilon}\{\cdot\mid \eta(y)\}$. We demonstrate that this is not a necessary condition.

For clarity we focus on the case of a scalar parameter $\theta $ and
scalar summary $\eta ({y})$, i.e., $k_{\theta}=k_{\eta}=1$, but present an extension to the multivariate case in Section \ref{subsec:comp}. This result requires a further assumption on the prior in addition to Assumption \ref{[A2]}.

\begin{assumption}\label{[A7]} The prior density $\pi(\theta)$ is such that
{(i)} for $\theta _{0}$ in the interior of ${\Theta }$, $%
\pi(\theta _{0})>0$; {(ii)} the density function 
$\pi(\theta)$ is $\beta $- H\"{o}lder in a neighbourhood of $\theta _{0}$: {%
there} exist $\delta ,L>0$ such that for all $|\theta -\theta
_{0}|\leq \delta $, and $\nabla_{\theta}^{(j)}\pi(\theta _{0})$ the $(j)$-th derivative of $\pi(\theta_0)$,
\begin{equation*}
\big\vert \pi(\theta )-\sum_{j=0}^{\lfloor \beta\rfloor }(\theta -\theta
_{0})^{j}\frac{\nabla_{\theta}^{(j)}\pi(\theta _{0})}{j!}\big\vert \leq L|\theta -\theta
_{0}|^{\beta }.
\end{equation*}%
{(iii)} For $\Theta \subset \mathbb{R}$, $\int_{\Theta
}|\theta |^{\beta }\pi(\theta )d\theta <\infty $.
\end{assumption}

\begin{theorem}\label{mean_thm} Let Assumptions \ref{[A2]},
\ref{[A1']}--\ref{[A5]},  with $\kappa >\beta +1$, and  \ref{[A7]}
be satisfied. Furthermore, let $\theta\mapsto b(\theta)$ be $\beta $-H\"{o}lder in a neighbourhood of $\theta_{0}$. Denoting $E_{\Pi
_{\varepsilon }}(\theta )$  as the posterior
mean of $\theta $, the following characterisation holds with probability going to one:

\noindent{(i)} If $\lim_{T}v_{T}\varepsilon _{T}=\infty $ and $v_{T}\varepsilon _{T}^{2\wedge (1+\beta )}=o(1)$, then 
\begin{equation}E_{\Pi
_{\varepsilon }}\left\{ v_{T}(\theta -\theta _{0})\right\} \rightarrow 
\mathcal{N}[0,V(\theta _{0})/\{\nabla _{\theta }b(\theta _{0})\}^{2}],
\label{normality:postmean}
\end{equation} in distribution as $T\rightarrow\infty$, 
where $V(\theta _{0})=\lim_{T}\text{\em var}[v_{T}\{\eta ({y})-b(\theta _{0})\}]$.

\noindent{(ii)} If $\lim_{T}v_{T}\varepsilon _{T}=c\geq 0$, and if when $c= 0$ Assumption \ref{[A6]} holds, then \eqref{normality:postmean} also holds.
\end{theorem}

There are two immediate consequences of Theorem \ref{mean_thm}: first, part (i) of Theorem \ref{mean_thm} states that if one is only interested in obtaining accurate point estimators for $\theta_0$, all we require is a  tolerance $\varepsilon_{T}$ satisfying $v_{T}\varepsilon^2_{T}=o(1)$, which can significantly reduce the computational burden of approximate Bayesian computation; secondly, if one wants accurate point estimators of $\theta_0$ and accurate expressions of the uncertainty associated with this point estimate, we require $\varepsilon_{T}=o(1/v_{T})$. The first statement follows directly from part (i) of Theorem \ref{mean_thm}, while the second statement follows from part (ii) of Theorem \ref{mean_thm} and recalling that, from Theorem  \ref{normal_thm}, credible regions constructed from $\Pi_{\varepsilon}\{\cdot\mid \eta(y)\}$ will have proper frequentist coverage only if $\varepsilon_{T}=o(1/v_{T})$. For $\varepsilon_{T}\asymp v_T^{-1}$ or $\varepsilon_{T}\gg v_{T}^{-1}$, the frequentist coverage of credible balls centered at $E_{\Pi_\varepsilon}(\theta)$ will not be equal to the nominal level.

As an intermediate step in the proof of Theorem \ref{mean_thm}, we demonstrate the following expansion for the posterior mean, 
where $k$ denotes the integer part of $(\beta+1)/2$:\begin{equation}
E_{\Pi _{\varepsilon }}( \theta -\theta _{0}) =\frac{\eta(y)-b(\theta_0)}{\nabla _{\theta }b(\theta _{0})}+\sum_{j=1}^{\lfloor \beta \rfloor }%
\frac{\nabla_{b}^{(j)}b^{-1}(b_{0})}{j!}\sum_{l=\lceil j/2\rceil }^{\lfloor
(j+k)/2\rfloor }\frac{\varepsilon _{T}^{2l}\nabla_{b}^{(2l-j)}\pi(b_{0})}{\pi(b_{0})(2l-j)!%
}+O(\varepsilon _{T}^{1+\beta })+o_{P}(1/v_{T})  \label{expan:postmean}.
\end{equation} 
This highlights a potential deviation from the expected asymptotic behaviour of the posterior
mean $E_{\Pi _{\varepsilon }}(\theta )$, i.e., the behaviour corresponding
to $T\rightarrow \infty $ and $\varepsilon _{T}\rightarrow0$. Indeed, the
posterior mean is asymptotically normal for all values of $\varepsilon _{T}=o(1)$%
, but is asymptotically unbiased only if the leading term in equation %
\eqref{expan:postmean} is $[\nabla _{\theta }b(\theta
_{0})]^{-1}\{\eta(y)-b(\theta_0)\} $, which is satisfied under Case (ii) and in Case (i) if $v_{T}\varepsilon _{T}^{2}=o(1)$, given $%
\beta \geq 1$. However, in Case (i), if $\lim\inf_{T}v_{T}\varepsilon _{T}^{2}>0$, when $\beta \geq 3$, the posterior mean has a bias
\begin{equation*}
\varepsilon _{T}^{2}\left[ \frac{\nabla_{b}\pi(b_{0})}{3\pi(b_{0})\nabla _{\theta }b(\theta
_{0})}-\frac{\nabla^{(2)}_{\theta}b(\theta _{0})}{2\{\nabla _{\theta }b(\theta
_{0})\}^{2}}\right] +O(\varepsilon _{T}^{4})+o_{P}(1/v_{T}).
\end{equation*}%

\subsection{Comparison with Existing Results}\label{subsec:comp}
 
Li and Fearnhead (2018a)  analyse the asymptotic
properties of the posterior mean and functions thereof. Under the assumption of a central limit theorem for the summary
statistic and further regularity assumptions on the convergence of the
density of the summary statistics to this normal limit, including the
existence of an Edgeworth expansion with exponential controls on the tails,
Li and Fearnhead (2018a) demonstrate asymptotic normality, with no bias, of the posterior
mean if $\varepsilon_T = o (1/v_{T}^{3/5})$. Heuristically, {the authors derive
this result using an approximation} of the posterior density $\pi_{\varepsilon
}\{\theta \mid \eta ({y})\}$, based on the Gaussian approximation of the density
of $\eta ({z})$ given $\theta$ and using properties of the maximum
likelihood estimator conditional on $\eta ({y})$. In contrast to our
analysis, these authors allow the acceptance probability defining
the algorithm to be an arbitrary density kernel in $\|\eta(y)-\eta(z)\|$.
Consequently, their approach is more general than the accept/reject version
considered in Theorem \ref{mean_thm}. 

However, the conditions Li and Fearnhead (2018a) require of $\eta ({y})$ are stronger than ours. In particular, our results on
asymptotic normality for the posterior mean only require weak convergence of 
$v_{T}\{\eta ({z}) - {b}(\theta)\}$ under $P_\theta$, with polynomial
deviations that need not be uniform in $\theta$. These assumptions
allow for the explicit treatment of models where the parameter space $%
\Theta$ is not compact. In addition,
asymptotic normality of the posterior mean requires Assumption \ref{[A6]} only if $\varepsilon_T = o(1/v_T)$. Hence if $\varepsilon_T \gg v^{-1}_T$, then only deviation bounds and weak convergence are required, which are much weaker than convergence of the densities. When $\varepsilon_T = o(1/v_T)$ then 
 Assumption \ref{[A6]} essentially implies local (in $\theta$) convergence of the
density of $v_{T}\{\eta ({z}) - {b}(\theta)\}$, but with no requirement
on the rate of this convergence. This assumption is weaker
than the uniform convergence required in Li and Fearnhead (2018a). Our
results also allow for an explicit representation of the bias that obtains for
the posterior mean when $\lim\inf_{T}v_{T}\varepsilon^{2}_T>0$.

In further contrast to Li and Fearnhead (2018a), Theorem \ref{normal_thm} completely characterizes the asymptotic behavior of the approximate posterior distribution for all $\varepsilon _{T}=o(1)$ that admit posterior
concentration. This general characterization allows us to
demonstrate, via Theorem \ref{mean_thm} part {(i)}, that asymptotic
normality and unbiasedness
of the posterior mean remain achievable even if $\lim_{T}v_{T}\varepsilon
_{T}=\infty $, provided the tolerance satisfies $\varepsilon _{T}=o(1/v_{T}^{1/2})$.

Li and Fearnhead (2018a) provide the interesting result that if $k_\eta >
k_\theta\geq1$ and $ \varepsilon_{T}=o(1/v_T^{3/5})$,  the posterior mean is
asymptotically normal, and unbiased, but is not asymptotically efficient. To help shed light
on this phenomenon, the following result gives an alternative to
Theorem 3.1 of these authors and contains an explicit
 asymptotic expansion for the posterior mean when $k_\eta >
k_\theta\geq1$.
\begin{theorem}\label{multivariatemean}
Let Assumptions \ref{[A2]}, \ref{[A1']}--\ref{[A5]} and \ref{[A7]} be satisfied. Assume that
$v_T\varepsilon_T \rightarrow \infty $ and $v_T\varepsilon_T^2 = o(1) $. 
Assume also that $b(.) $ and $\pi(.)$   are Lipschitz in a neighbourhood of $\theta
_{0}$.  Then, for $k_\eta >
k_\theta\geq1$,
$$
E_{\Pi_{\varepsilon}}\{ v_T(\theta-\theta_0) \} 
=  \{\nabla_\theta b(\theta_0)^{\intercal} \nabla_\theta b(\theta_0)\}^{-1}
\nabla_\theta b(\theta_0)^{\intercal}v_{T}\{\eta(y)-b(\theta_0)\}+ o_p(1).
$$ 
In addition, if $\{\nabla_\theta b(\theta_0)^{\intercal} \nabla_\theta
b(\theta_0)\}^{-1}\nabla_\theta b(\theta_0)^{\intercal} \neq  \nabla_\theta
b(\theta_0)^{\intercal}$, the matrix
$$
\text{\em var}\left[\{\nabla_\theta b(\theta_0)^{\intercal}\nabla_\theta
b(\theta_0)\}^{-1}\nabla_\theta b(\theta_0)^{\intercal}v_{T}\{\eta(y)-b(\theta_0)\}\right]-
\{\nabla_\theta b(\theta_0)^{\intercal}V^{-1}(\theta_0)\nabla_\theta
b(\theta_0)^{}\}^{-1},
$$  
is positive semi-definite, where $\{\nabla_\theta
b(\theta_0)^{\intercal}V^{-1}(\theta_0)\nabla_\theta b(\theta_0)^{}\}^{-1}$ is the
optimal asymptotic variance achievable given $\eta(y)$.
\end{theorem}

A consequence of Theorem \ref{multivariatemean} is that, for a fixed choice of summaries,
the two-stage procedure advocated by Fearnhead and Prangle (2012) will not reduce the asymptotic variance over a point estimate produced via Algorithm 1. However, this two-stage procedure does reduce the Monte Carlo error inherent in estimating the approximate posterior distribution $\Pi_{\varepsilon}\{\cdot\mid \eta(y)\}$ by reducing the dimension of the statistics on which the matching in approximate Bayesian computation is based.

\section{Practical Implications of the Results\label{pract}}

\subsection{General \label{sec:implications}}

The approximate Bayesian computation approach in Algorithm 1 is
typically not applied in practice. Instead, the acceptance step in Algorithm
1 is commonly replaced by the nearest-neighbour selection step and with $d_{2}\{\eta(z),\eta(y)\}=\|\eta(z)-\eta(y)\|$, see,
e.g., Biau {et al.}
(2015):
\smallskip

\noindent (3$'$) select all $\theta ^{i}$ associated with the $%
\alpha =\delta /N$ smallest distances $\|\eta(z)-\eta(y)\|$ for some $%
\delta $. \smallskip

This nearest-neighbour version accepts draws of $\theta $
associated with an empirical quantile over the simulated distances $\|\eta(z)-\eta(y)\|$ and defines the acceptance probability
for Algorithm 1. A key practical insight of our asymptotic results is
that the acceptance probability, 
 $\alpha _{T}=\text{pr}\left\{\|\eta(z)-\eta(y)\|\leq
\varepsilon _{T}\right\}$, is only affected by the dimension of $\theta $,
as formalized in Corollary \ref{implication:accept}.

\begin{corollary}
\label{implication:accept} Under the conditions in Theorem \ref{normal_thm}:

\noindent{(i)} If $\varepsilon _{T}\asymp v_{T}^{-1}$ or $\varepsilon
_{T}=o(1/v_{T})$, then the acceptance rate associated with the threshold $%
\varepsilon _{T}$ is 
\begin{equation*}
\alpha _{T}=\text{\em pr}\left\{\|\eta(z)-\eta(y)\| \leq
\varepsilon _{T}\right\} \asymp (v_{T}\varepsilon _{T})^{k_{\eta }}\times
v_{T}^{-k_{\theta }}\lesssim v_{T}^{-k_{\theta}}.
\end{equation*}

\noindent{(ii)} If $\varepsilon _{T}\gg v_{T}^{-1}$ , then 
\begin{equation*}
\alpha _{T}=\text{\em pr}\left\{\|\eta(z)-\eta(y)\|\leq
\varepsilon _{T}\right\} \asymp \varepsilon _{T}^{k_{\theta }}\gg v_{T}^{-k_{\theta}}.
\end{equation*}
\end{corollary}

This shows that choosing a tolerance $%
\varepsilon _{T}=o(1)$\ is equivalent to choosing an $\alpha _{T}=o(1)$\
quantile of $\|\eta(z)-\eta(y)\| $. It
also demonstrates the role played by the dimension of $\theta $ on
the rate at which $\alpha _{T}$\ declines to zero. In Case {(i)}, if $%
\varepsilon _{T}\asymp v_{T}^{-1}$, then $\alpha _{T}\asymp
v_{T}^{-k_{\theta }}.$ On the other hand, if $\varepsilon _{T}=o(1/v_{T})$%
, as required for the Bernstein--von Mises result in Theorem \ref{normal_thm}, the associated acceptance probability goes to zero at the
faster rate, $\alpha _{T}=o(1/v_{T}^{k_{\theta }}).$ In Case {(ii)}, where $\varepsilon _{T}\gg v_{T}^{-1}$, it follows that $\alpha _{T}\gg
v_{T}^{-k_{\theta }}$.

Linking $\varepsilon _{T}$ and $\alpha _{T}$ gives a means of
choosing the $\alpha _{T}$ quantile of the simulations, or equivalently the
tolerance $\varepsilon _{T}$, in such a way that a particular type of
posterior behaviour occurs for large $T$: choosing $%
\alpha _{T}\gtrsim v_{T}^{-k_{\theta }}$ gives an approximate posterior distribution that
concentrates; under the more stringent condition $\alpha
_{T}=o(1/v_{T}^{k_{\theta }})$ the approximate posterior distribution both concentrates {and%
} is approximately Gaussian in large samples. These results give practitioners an understanding of what to expect from this
procedure, and a means of detecting potential issues if this expected behaviour is not in evidence. Moreover, given that there is
no direct link between $\Pi_{\varepsilon}\{\cdot\mid \eta(y)\}$
and the exact posterior distribution, these results give some understanding of the
statistical properties that $\Pi_{\varepsilon}\{\cdot\mid \eta(y)\}$ should display when it is
obtained from the popular nearest-neighbour version of the algorithm.

Corollary \ref{implication:accept}\ demonstrates that to obtain reasonable statistical behavior, the rate at which $\alpha _{T}$\ declines to zero
must be faster the larger the dimension of $\theta $, with the order of $\alpha _{T}$ unaffected by the dimension of $\eta$. This result provides theoretical evidence of a
curse-of-dimensionality encountered in these algorithms as the dimension of
the parameters increases, with this being the first piece of
work, to our knowledge, to link the dimension of $\theta $ to  certain asymptotic properties for $\Pi_{\varepsilon}\{\cdot\mid \eta(y)\}$. This result provides theoretical justification for dimension reduction methods that
process parameter dimensions individually and independent of the other
dimensions; see, for example, the regression adjustment approaches
of Beaumont et al. (2002), Blum (2010) and Fearnhead and Prangle (2012), and
the integrated auxiliary likelihood approach of Martin et al. (2016).

While Corollary \ref{implication:accept} demonstrates that the order of $\alpha_{T}$ is unaffected by the dimension of the summaries, $\alpha_{T}$ cannot be accessed in practice and so the nearest-neighbour version of Algorithm 1 is implemented using a Monte Carlo approximation to $\alpha_{T}$, which is based on the accepted draws of $\theta$. {This approximation of $\alpha_{T}$ is a Monte Carlo estimate of a conditional expectation, and, as such, will be sensitive to the dimension of ${\eta}(\cdot)$ for any fixed number of Monte Carlo draws $N$; see Biau et al. (2015) for further discussion on this point. In addition,} it can also be shown that if
$\varepsilon_T$ becomes much smaller than $1/v_T$, the dimension of $\eta(\cdot)$ will affect the behavior of Monte Carlo estimators for this
acceptance probability. Specifically, when considering inference on $\theta_0$ using the accept/reject approximate Bayesian computation algorithm, we require a sequence of Monte Carlo trials $N_{T}\rightarrow\infty$ as $T\rightarrow\infty$ that diverges faster the larger is $k_{\eta}$, the dimension of $\eta(\cdot)$. Such a feature highlights the lack of efficiency of the accept/reject approach when the sample size is large or if the dimension of the summaries is large. However, we note here that more efficient sampling approaches exist and could be applied in these settings. For example, Li and Fearnhead (2018a) consider an importance sampling approach to approximate Bayesian computation that yields acceptance rates satisfying $\alpha_{T}=O(1)$, so long as $\varepsilon_{T}=O(1/v_{T})$. Therefore, in cases where the Monte Carlo error is likely to be large, these alternative sampling approaches should be employed. 

Regardless of whether one uses a more efficient sampling procedure than the simple accept/reject approach, Corollary \ref{implication:accept} demonstrates that taking a tolerance sequence as small as possible will not necessarily yield more accurate results. That is,
Corollary \ref{implication:accept} questions the persistent opinion that the tolerance in Algorithm 1 should always
be taken  as small as the computing budget
allows. Once $\varepsilon _{T}$
is chosen small enough to satisfy Case {(iii)} of Theorem \ref%
{normal_thm}, which leads to the most stringent requirement on the tolerance, $%
v_{T}\varepsilon_{T}=o(1)$, there may well be no gain in pushing $%
\varepsilon _{T}$ or, equivalently, $\alpha _{T}$ any closer to
zero, {especially since pushing $\varepsilon_{T}$ closer to zero can drastically increase the required computational burden}. {In the following section we numerically demonstrate this result in a simple example. In particular, we demonstrate that for a choice of tolerance $\varepsilon_{T}$ that admits a Bernstein--von Mises result, there is no gain in taking a tolerance that is smaller than this value, while the computational cost associated with such a choice, for a fixed level of Monte Carlo error, drastically increases. }

\subsection{Numerical Illustration of Quantile Choice}

Consider the simple example where we observe a sample $\{y_{t}\}_{t=1}^{T}$
from $y_{t}\sim \mathcal{N}(\mu ,\sigma )$ with $T=100$. Our goal is
posterior inference on $\theta=(\mu ,\sigma)^{\intercal}$. We use as 
summaries the sample mean and variance, $\bar{x}$ and $s_{T}^{2}$, which
satisfy a central limit theorem at rate ${T}^{1/2}$. In order 
to guarantee asymptotic normality of the approximate posterior distribution, we must choose an $\alpha _{T}$
quantile of the simulated distances according to $\alpha _{T}=o(1/T)$, 
because of the joint inference on $\mu $ and $\sigma $. For the purpose
of this illustration, we will compare inference based on the nearest-neighbour version of Algorithm 1
using four different choices of $\alpha _{T}$, $\alpha _{1}=1/T^{1.1},$ $\alpha
_{2}=1/T^{3/2}$, $\alpha _{3}=1/T^{2}$ and $\alpha _{4}=1/T^{5/2}$.

Draws for $(\mu ,\sigma )$ are simulated on $[0.5,1.5]\times \lbrack
0.5,1.5] $ according to independent uniforms $\mathcal{U}[0.5,1.5]$. The
number of draws, $N$, is chosen so that we retain 250 accepted
draws for each of the different choices ($\alpha _{1},\dots,\alpha _{4}$). The
exact finite sample marginal posterior densities of $\mu $\ and $\sigma $ are
produced by numerically evaluating the likelihood function, normalizing over
the support of the prior and marginalising with respect to each parameter.
Given the sufficiency of $(\bar{x},\;s_{T}^{2})$, the exact marginal
posteriors densities for $\mu$ and $\sigma$ are equal to those based directly on the
summaries themselves. Hence, we are able to assess the impact of the choice of $\alpha$, per se, on the ability of the nearest-neighbour version of Algorithm 1 to replicate the exact marginal posteriors. 

We summarize the accuracy of the resulting approximate posterior density estimates, across the
four quantile choices, using root mean squared error. In particular, over fifty simulated replications, and in the case of the parameter $\mu $, we estimate
the root mean squared error between the marginal posterior density obtained from
Algorithm 1 using $\alpha _{j}$, and denoted by $\widehat{\pi}_{\alpha
_{j}}\{\mu \mid{\eta }({y})\}$, and the exact marginal posterior density, $\pi(\mu \mid {y})$,
using 
\begin{equation}
\textsc{RMSE}_{\mu }(\alpha _{j})=\left[\frac{1}{G}\sum_{g=1}^{G}\big\{ \widehat{\pi}%
_{\alpha _{j}}^{g}\{\mu \mid {\eta }({y})\}-\pi^{g}(\mu \mid {y})\big\}
^{2}\right]^{1/2}. \label{RMSE_def}
\end{equation}%
The term $\widehat{\pi}^{g}_{\alpha_{j}}$ is the ordinate of the density estimate from the nearest-neighbour version of
Algorithm 1 and $\pi^{g}$ is the ordinate of the exact posterior density,
at the $g$-th grid point upon which the density is estimated. $\textsc{RMSE}_{\sigma }(\alpha _{j})$ is computed analogously. The value of $\textsc{RMSE}_{\mu }(\alpha _{j})$ is averaged over
fifty replications to account for sampling variability. For each replication, we fix $T=100$ and generate observations using the
parameter values $\mu _{0}=1,\;\sigma _{0}=1$.

Before presenting the replication results, it is instructive to
consider the graphical results of one particular run of the algorithm for
each of the $\alpha _{j}$ values. Figure \ref{fig1} plots the resulting
marginal posterior estimates and compares these with the exact finite
sample marginal posterior densities of $\mu $ and $\sigma $. At the end of Section \ref{sec:implications}, we argued that for large enough $T$, once $\varepsilon _{T}$ reaches a certain
threshold, decreasing the tolerance further will not necessarily result in
more accurate estimates of these exact posterior densities. This implication is
 evident in Fig.~\ref{fig1}: in the case of $\mu $, there is a
clear visual decline in the accuracy with which approximate Bayesian computation estimates the exact
marginal posterior densities when choosing quantiles smaller than $\alpha _{2}$;{\ }%
whilst in the case of $\sigma $, the worst performing estimate is that
associated with the smallest value of $\alpha _{j}.$%
\begin{figure}[tbp]
\centering 
\setlength\figureheight{2.25cm} 
\setlength\figurewidth{6.0cm} 
\input{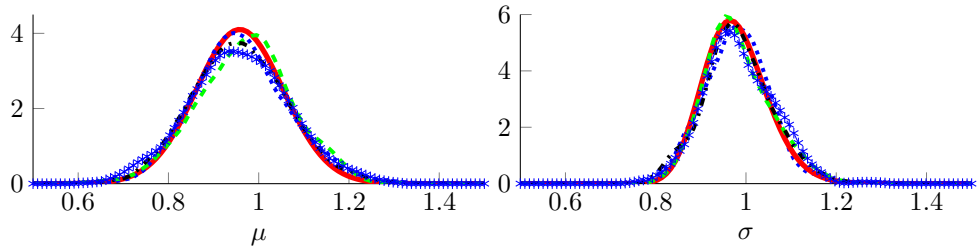}
\caption{Comparison of exact and approximate posterior densities for various tolerances. Exact marginal posterior densities ({\color{red}\textbf{\----}}). Approximate Bayesian computation posterior densities based on $\alpha_{1}=1/T^{1.1.}$ ({\color{blue}$\cdots$}); $\alpha_{2}=1/T^{3/2}$ (\textbf{\color{green}- - -}); $\alpha_{3}=1/T^{2}$ (\textbf{\color{black}\---- $\cdot$ \----}); $\alpha_{4}=1/T^{5/2}$ {\color{blue}\----$*$\----}}
\label{fig1}
\end{figure}

The results in Table \ref{tab1} report average root mean squared error, relative
to the average value associated with $\alpha
_{4}=1/T^{5/2}.$ Values smaller than one indicate that the larger, and
less computationally burdensome, value of $\alpha _{j}$ yields a more accurate estimate than that obtained using $\alpha
_{4} $. In brief, Table \ref{tab1} paints a similar picture to that of
Fig.~\ref{fig1}: for $\sigma $, the estimates based on $\alpha _{j}$, $
j=1,2,3,$ are all more accurate than those based on $\alpha
_{4} $; for $\mu $, estimates based on $
\alpha _{2}$ and $\alpha _{3}$ are both more accurate that those based on $
\alpha _{4}$.

These numerical results have important implications for implementation of approximate Bayesian computation. In particular, to keep the level of Monte Carlo error constant across the $\alpha_{j}$ quantile choices, as we have done in this simulation setting via the retention of 250 draws, this requires taking: $N=210e03$ for $\alpha _{1}$, $N=1.4e06$ for $%
\alpha _{2}$, $N=13.5e06$ for $\alpha _{3}$, and $N=41.0e06$
for $\alpha _{4}$. That is, the computational burden associated
with decreasing the quantile in the manner indicated increases dramatically: approximate
posterior densities based on $\alpha _{4}$\ for example require a value of $N$ that
is three orders of magnitude greater than those based on $\alpha _{1}$, but
this increase in computational burden yields no, or minimal, gain in
accuracy. The extension of such explorations to more scenarios is beyond the
scope of this paper; however, we speculate that, with due consideration
given to the properties of both the true data generating process and the
chosen summary statistics and, hence, of the sample sizes for which Theorem %
\ref{normal_thm} has practical content, similar qualitative results
will continue to hold.

\begin{table}[tbph]
\def~{\hphantom{0}}
\tbl{Ratio of the average root mean square
error for marginal approximate posterior density estimates relative to the average root mean square
error based on the smallest {quantile}, $\protect\alpha_{4}=1/T^{5/2}$}{
\begin{tabular}{rrrrrrr}
& $\alpha _{1}=1/T^{1.1}$ &  & $\alpha _{2}=1/T^{1.5}$ &  & $\alpha
_{3}=1/T^{2}$ &  \\ &&&&&&\\
$\textsc{AVG-RMSE}_{\mu }(\alpha _{j})$ & 1.17 &  & 0.99 &  & 0.98 &  \\ 
$\textsc{AVG-RMSE}_{\sigma }(\alpha _{j})$ & 0.86 &  & 0.87 &  & 0.91 &  \\ 
\end{tabular}}%
\begin{tabnote}
\textsc{AVG-RMSE} is the ratio of the average root mean square errors as defined in \eqref{RMSE_def}.
\label{tab1}
\end{tabnote}
\end{table}

\section*{Acknowledgment}

This research was supported by Australian Research Council Discovery and
l'Institut Universitaire de France grants. The third author is further
affiliated with the University of Warwick and the fourth with the Universit\'e
Paris Dauphine.  We are grateful to Wentao Li and Paul Fearnhead for spotting a
mistake in a previous version of Theorem \ref{mean_thm} and to the editorial
team for its help.

%

\bigskip

\Large
\begin{center}
{\bf Asymptotic Properties of Approximate Bayesian Computation:
Supplementary Material}
\end{center}
\normalsize

\bigskip

\begin{abstract}
This supplementary material contains proofs of Theorems 1--4 {and Corollary 1%
} in the paper. In addition, we illustrate the implications of Theorems 1--3
in the paper with a series of simulated examples {based on} {the }%
moving average model {of Example 1.}
\end{abstract}

\section{Proofs}

\label{sec:pf}

\subsection{Proof of Theorem 1}

Let $\varepsilon _{T}>0$, where, by assumption $\varepsilon_{T}=o(1)$, 
and assume that ${y}\in \Omega
_{\varepsilon }=\{{y}:d_{2}\{{\eta }({y}),{b}({\theta }_{0})\}\leq
\varepsilon _{T}/3\}$. From assumption \ref{[A1]} and $\rho _{T}(\varepsilon
_{T}/3)=o(1)$, $P_{0}\left( \Omega _{\varepsilon }\right) =1+o(1)$. Consider
the joint event $A_{\varepsilon }(\delta ^{\prime })=\{({z},{\theta }%
):d_{2}\{{\eta }({z}),{\eta }({y})\}\leq \varepsilon _{T}\}\cap d_{2}\{{b}({%
\theta }),{b}({\theta }_{0})\}>\delta ^{\prime }\}$. For all $({z},{\theta }%
)\in A_{\varepsilon }(\delta ^{\prime })$ 
\begin{equation*}
\begin{split}
d_{2}\{{b}({\theta }),{b}({\theta }_{0})\}& \leq d_{2}\{{\eta }({z}),{\eta }(%
{y})\}+d_{2}\{{b}({\theta }),{\eta }({z})\}+d_{2}\{{b}({\theta }_{0}),{\eta }%
({y})\} \\
& \leq 4\varepsilon _{T}/3+d_{2}\{{b}({\theta }),{\eta }({z})\}.
\end{split}%
\end{equation*}%
{Hence} $({z},{\theta })\in A_{\varepsilon }(\delta ^{\prime })$
implies that 
\begin{equation*}
d_{2}\{{b}({\theta }),{\eta }({z})\}>\delta ^{\prime }-4\varepsilon _{T}/3
\end{equation*}%
and choosing $\delta ^{\prime }\geq 4\varepsilon _{T}/3+t_{\varepsilon }$
leads to%
\begin{equation*}
\text{pr}\left\{ A_{\varepsilon }(\delta ^{\prime })\right\} \leq
\int_{\Theta }P_{{\theta }}\left[d_{2}\{{b}({\theta }),{\eta }({z}%
)\}>t_{\varepsilon }\right] d\Pi ({\theta }),
\end{equation*}%
and 
\begin{flalign}
\Pi [d_{2} \{{b}({\theta }),{b}(
{\theta }_{0})\}&>4\varepsilon _{T}/3+t_{\varepsilon }\mid d_{2}\{
{\eta }({y}),{\eta }({z})\}\leq
\varepsilon _{T}] =\Pi _{\varepsilon }\left[ d_{2}\{{b}(
{\theta }),{b}({\theta }_{0})\}>4
\varepsilon _{T}/3+t_{\varepsilon }\mid{\eta(y) }_{}\right] 
\notag \\
&\leq {\int_{\Theta }P_{\theta }\left[ d_{2}\{{b}(
{\theta }),{\eta }({z})\}>t_{\varepsilon
}\right] d\Pi ({\theta })}\Big/{\int_{\Theta }P_{\theta }\left[ d_{2}\{
{\eta }({z}),{\eta }({y})\}\leq
\varepsilon _{T}\right] d\Pi ({\theta })}.  \label{ineq1}
\end{flalign}Moreover, since 
\begin{equation*}
d_{2}\{{\eta }({z}),{\eta }({y})\}\leq d_{2}\{{b}({\theta }),{\eta }({z}%
)\}+d_{2}\{{b}({\theta }_{0}),{\eta }({y})\}+d_{2}\{{b}({\theta }),{b}({%
\theta }_{0})\}\leq \varepsilon _{T}/3+\varepsilon _{T}/3+d_{2}\{{b}({\theta 
}),{b}({\theta }_{0})\},
\end{equation*}%
provided $d_{2}\{{b}({\theta }),{\eta }({z})\}\leq \varepsilon _{T}/3$, then 
\begin{equation*}
\begin{split}
\int_{\Theta }P_{{\theta }}\left[ d_{2}\{{\eta }({z}),{\eta }({y})\}\right.
& \leq \left. \varepsilon _{T}\right] d\Pi ({\theta })\geq \int_{d_{2}\{{b}({%
\theta }),{b}({\theta }_{0})\}\leq \varepsilon _{T}/3}P_{{\theta }}\left[
d_{2}\{{\eta }({z}),b(\theta )\}\leq \varepsilon _{T}/3\right] d\Pi ({\theta 
}) \\
& \geq \Pi \left[d_{2}\{{b}({\theta }),{b}({\theta }_{0})\}\leq \varepsilon
_{T}/3\right] -\rho _{{T}}(\varepsilon _{T}/3)\int_{d_{2}\{{b}({\theta }),{b}%
({\theta }_{0})\}\leq \varepsilon _{T}/3}c({\theta })d\Pi ({\theta }).
\end{split}%
\end{equation*}%
If part {(i)} of assumption \ref{[A1]} holds, 
\begin{equation*}
\int_{d_{2}\{{b}({\theta }),{b}({\theta }_{0})\}\leq \varepsilon _{T}/3}c({%
\theta })d\Pi ({\theta })\leq c_{0}\Pi \left[ d_{2}\{{b}({\theta }),{b}({%
\theta }_{0})\}\leq \varepsilon _{T}/3\right] 
\end{equation*}%
and for $\varepsilon _{T}$ small enough, or for $T$ large enough, 
so that $\rho _{T}(\varepsilon _{T}/3)$ is {small}, 
$$
\int_{\Theta }P_{{\theta }}\left[d_{2}\{{\eta }({z}),{\eta }({y})\}\leq
\varepsilon _{T}\right] d\Pi ({\theta })\geq {\Pi \left[ d_{2}\{{b}({\theta }%
),{b}({\theta }_{0})\}\leq \varepsilon _{T}/3\right] }/{2},  
$$
which, combined with \eqref{ineq1} and assumption \ref{[A2]}, leads to%
\begin{equation}
\Pi\left[ d_{2}\{{b}({\theta }),{b}({\theta }%
_{0})\}>4\varepsilon _{T}/3+t_{\varepsilon }\mid d_{2}\{\eta(z),\eta(y)\}\leq \varepsilon_{T}\right] \lesssim
\rho _{T}(t_{\varepsilon })\varepsilon _{T}^{-D}\lesssim {1}/{M}
\label{ineq-i}
\end{equation}%
by choosing $t_{\varepsilon }=\rho _{T}^{-1}(\varepsilon _{T}^{D}/M)$ with $M
$ large enough. If part {(ii)} of assumption \ref{[A1]} holds, a H\"{o}lder
inequality implies that 
\begin{equation*}
\int_{d_{2}\{{b}({\theta }),{b}({\theta }_{0})\}\leq \varepsilon _{T}/3}c({%
\theta })d\Pi ({\theta })\lesssim \Pi \left[d_{2}\{{b}({\theta }),{b}({%
\theta }_{0})\}\leq \varepsilon _{T}/3\right] ^{a/(1+a)}
\end{equation*}%
and if $\varepsilon _{T}$ satisfies 
\begin{equation*}
\rho _{T}(\varepsilon _{T})=o\left\{ \varepsilon _{T}^{D/(1+a)}\right\}
=O\left(\Pi \left[ d_{2}\{{b}({\theta }),{b}({\theta }_{0})\}\leq
\varepsilon _{T}/3\right] ^{1/(1+a)}\right),
\end{equation*}%
then \eqref{ineq-i} remains valid. 

\subsection{Generalization of Theorem 2 and its Proof}

We obtain a generalization of Theorem 2 that allows differing rates of convergence for $\eta (y)$.
We assume here that there exists a sequence of $%
k_{\eta }\times k_{\eta }$ positive definite matrices ${\Sigma }_{T}({\theta 
})$ such that for all ${\theta }$ in a neighbourhood of ${\theta }_{0}$,
where $\theta_0$ is in  the interior of $\Theta$, 
\begin{equation}
c_{1}{D}_{T} \leq \Sigma _{T}({\theta })\leq c_{2} {D}_{T},\quad {D}_{T}=%
\mbox{diag}\{v _{T}(1),\ldots ,v_{T}(k)\},  \label{sigma_T}
\end{equation}%
with $0<c_{1},c_{2}<\infty $, $v_{T}(j)\rightarrow \infty $ for all $j$ and
the $v_{T}(j)$ are possibly all distinct. For square matrices $A,B$, 
$A\leq B$ means that the matrix $B-A$ is positive semi-definite. Thus,
this generalization of Theorem 2 does not require identical convergence rates for the components of the
statistic ${\eta }({z}) $. For simplicity, we order the components so that 
\begin{equation}
v_{T}(1)\leq \ldots \leq v_{T}(k_{\eta }).  \label{d_values}
\end{equation}%
%
%
For any square matrix ${A}$ of dimension $k_{\eta }$, if $q\leq k_{\eta }$, $%
{A}_{[q]}$ denotes the $q\times q$ square upper sub-matrix of ${A}$. Also,
let $j_{\max }=\max \{j:\lim_{T\to\infty}v_{T}(j)\varepsilon _{T}=0\}$ and
if, for all $j$, $\lim_{T\to\infty}v_{T}(j)\varepsilon _{T}>0$ then $j_{\max
}=0$.

In addition to assumption \ref{[A2]} in Section 3 of the text, the following
conditions are needed to establish this generalization of Theorem 2.

\medskip

\begin{assumption}
\label{[A1"]} Assumption 1 holds and
 the sequence of positive definite matrices $\{\Sigma _{T}({\theta }_{0})\}_{T\geq1}$ in \eqref{sigma_T} exists. For $\kappa >1$ and $\delta >0$, such that for all $\Vert {\theta }-{%
\theta }_{0}\Vert \leq \delta $, $P_{{\theta }}\left[ \Vert \Sigma _{T}({%
\theta }_{0})\{{\eta }({z})-{b}({\theta })\}\Vert >u\right] \leq {c_{0}
}/{u^{\kappa }}$ for all $0<u\leq \delta v_{T}(1)$ and $c_{0}<\infty$.\end{assumption}

\begin{assumption}\label{[A3"]} Assumption \ref{[A3]} holds, the
function ${b}(\cdot )$ is continuously differentiable at ${\theta
_{0}}$, and the Jacobian $\nabla _{\theta }b({\theta }_{0})$ has full column rank $%
k_{\theta }$.\end{assumption} 

\begin{assumption}\label{[A4"]} Given the sequence of $k_{\eta }\times k_{\eta }$ positive definite matrices ${%
\Sigma }_{T}({\theta })$ defined in (\ref{sigma_T}), for some $\delta >0$ and all $\Vert {\theta }-{\theta }_{0}\Vert \leq \delta $, the convergence in distribution: for all open sets $B$
\begin{equation*}
\sup_{|\theta - \theta_0|\leq \delta} \left| P_\theta \left[{\Sigma }_{T}({\theta })\{{\eta }({z})-{b}({\theta })\} \in B\right] - P\left\{ \mathcal{N%
}(0,I_{k_{\eta }}) \in B\right\} \right|  \longrightarrow 0 
\end{equation*}%
holds, where $I_{k_{\eta }}$ is the $k_{\eta }\times k_{\eta }$ identity
matrix.\end{assumption}

\begin{assumption}\label{[A5"]} For all $\Vert \theta -\theta
_{0}\Vert \leq \delta $, the sequence of functions ${\theta }\mapsto {\Sigma 
}_{T}({\theta })D_{T}^{-1}$ converges to some positive definite matrix $A(\theta )$
and is equicontinuous at ${\theta }_{0}$. \end{assumption}

\begin{assumption}\label{[A6"]} For some positive $\delta $ and
all $\Vert {\theta }-{\theta }_{0}\Vert \leq \delta $, and for all ellipsoids 
\begin{equation*}
B_{T}=\left\{(t_{1},\ldots ,t_{j_{\max }}):\tsum_{j=1}^{j_{\max }}\, t_{j}^{2}/h_{T}(j)^{2}\leq 1\right\}
\end{equation*}%
with $\lim_{T\to\infty} h_{T}(j) = 0$, for all $j\leq j_{\max }$ and all $u\in \mathbb{%
R}^{j_{\max }}$ fixed, 
\begin{equation*}
\begin{split}
\lim_{T\to\infty}\sup_{|\theta -\theta_0|\leq \delta } \left| \frac{P_{{\theta }}\left[ \{{\Sigma }_{T}({\theta })\}_{[j_{\max }]}\{{%
\eta }({z})-{b}({\theta })\}-u\in B_{T}\right] }{\prod_{j=1}^{j_{\max
}}h_{T}(j)} - \varphi _{j_{\max }}(u)\right| & =0, \\
\frac{P_{{\theta }}\left[ \{{\Sigma }_{T}({\theta })\}_{[j_{\max }]}\{{\eta }({%
z})-{b}({\theta })\}-u\in B_{T}\right] }{\prod_{j=1}^{j_{\max }}h_{T}(j)}&
\leq H(u),\quad \int H(u)du<\infty ,
\end{split}
\end{equation*}%
for $\varphi _{j_{\max}}(\cdot )$ the density of a $j_{\max }$-dimensional
normal random variate.
\end{assumption}

\setcounter{theorem}{5} 
\begin{theorem}
\label{normal_thm"} Assume that Assumptions 2, \ref{[A1"]}, with $\kappa >k_{\theta }$, and  \ref{[A3"]}--\ref{[A5"]}, are satisfied. The following results hold with $P_0$ probability approaching one: \medskip

\noindent {(i)} if $\lim_{T\to\infty}v_{T}(1)\varepsilon _{T}=\infty $, the posterior distribution of $\varepsilon
_{T}^{-1}({\theta }-{\theta }_{0})$ converges to the uniform distribution
over the ellipse $\{w :w^{\intercal}B_{0}w\leq 1\}$ with $B_{0}=\{\nabla _{\theta }{b}({%
\theta }_{0})^{^{\intercal }}\nabla _{\theta }{b}({\theta }_{0})\}$. Hence, for
all $f(\cdot)$ continuous and bounded,
\begin{equation}
\int f\{\varepsilon _{T}^{-1}({\theta }-{\theta }_{0})\}d\Pi _{\varepsilon
}\{\theta  \mid {\eta(y) }_{}\}{\rightarrow } {\int_{u^{\intercal}B_{0}u\leq
1}f(u)du}\big/{\int_{u^{\intercal}B_{0}u\leq 1}du}.  \label{aneninf"}
\end{equation}

\noindent {(ii)} if there exists $k_{0}<k_{\eta }$ such that $%
\lim_{T\to\infty}v_{T}(1)\varepsilon _{T}=\lim_{T\to\infty}v_{T}(k_{0})\varepsilon
_{T}=c $, $0<c<\infty $, and $\lim_{T\to\infty}v_{T}(k_{0}+1)\varepsilon
_{T}=\infty $, assuming
$$\text{\em Leb}\left( \sum_{j=1}^{k_{0}}\left[ \left\{ {\nabla _{\theta }{b}({\theta }
_{0})}({\theta }-{\theta _{0}})\right\} _{[j]}\right] ^{2}\leq c\varepsilon
_{T}^{2}\right) =\infty,$$ then 
$$
\Pi _{\varepsilon }\left[ {\Sigma }_{T}({\theta }_{0})\{{b}({\theta })-{b}({%
\theta }_{0})\}-\Sigma_{T}(\theta_0)\{\eta(y)-b(\theta_0)\}\in B \mid {\eta(y) }_{}\right] \rightarrow 0,
$$
for all bounded measurable sets $B$, where $\text{\em Leb}(.)$ denotes the Lebesgue measure.\medskip

\noindent {(iii)} if there exists $j_{\max }<k_{\eta }$ such that $%
\lim_{T\to\infty}v_{T}(j_{\max })\varepsilon _{T}=0$ and $\lim_{T\to\infty}v_{T}(j_{\max
}+1)\varepsilon _{T}=\infty $, if assumption \ref{[A6"]} is satisfied and  $\kappa$ is such that
$$\left\{\prod_{j=1}^{j_{\max }}v
_{T}(j) \right\}^{-1/(\kappa+ j_{\max} )} v_T(j_{\max }+1)^{-\kappa/(\kappa+j_{\max} )}  = o(\varepsilon _{T}),$$
 then \eqref{aneninf"} is satisfied.\medskip

\noindent {(iv)} if $\lim_{T\to\infty}v_{T}(j)\varepsilon _{T}=c>0$ for all $%
j\leq k_{\eta }$ or  if case (ii) holds with  
$$\text{ \em Leb}\left( \sum_{j=1}^{k_{0}}\left[ \left\{ {\nabla
_{\theta }{b}({\theta }_{0})}({\theta }-{\theta _{0}})\right\} _{[j]}\right]
^{2}\leq c\varepsilon _{T}^{2}\right) <\infty ,$$ 
then there exists a non-Gaussian probability distribution on $\mathbb{R}^{k_{\eta
}}$, $Q_{c}$, such that
$$
\Pi _{\varepsilon }\left[ {\Sigma }_{T}({\theta }_{0})\{{b}({\theta })-{b}({%
\theta }_{0})\}-\Sigma_{T}(\theta_0)\{\eta(y)-b(\theta_0)\}\in B \mid {\eta(y) }_{}\right] \rightarrow Q_{c}(B).
$$
In particular, 
\begin{equation*}
Q_{c}(B) \propto \int_{B}\int_{\mathbb{R}^{k_{\eta}} }\1_{(z-x)^{\intercal}A(\theta_0)^{\intercal}A(\theta_0)(z-x)}{\varphi}_{k_\eta}(z) dzdx.
\end{equation*}

\noindent {(v)} if $\lim_{T\to\infty}v_{T}(k_{\eta})\varepsilon _{T}=0$ and under assumption \ref{[A6"]}
holding for $j_{\max }=k_{\eta }$, then, for $\Phi_{k_{\eta}}(\cdot)$ the  cumulative distribution function of the $k_{\eta}$-dimensional standard normal 
$$
\lim_{T\to\infty}\Pi_{\varepsilon }\left[\Sigma_{T}(\theta_{0})\{{b}(\theta
)-b(\theta_{0})\}-\Sigma_{T}(\theta_0)\{\eta(y)-b(\theta_0)\}\in B \mid {\eta(y) }_{}\right] =\Phi _{k_{\eta}}({B}).
$$
\end{theorem}

\begin{proof}\renewcommand{\qedsymbol}{}
We work with ${b}(\theta)$ instead of ${\theta }$ as the parameter, with injectivity
of ${\theta }\mapsto {b}({\theta })$ required to re-state all results in
terms of ${\theta }.$ For mathematical convenience, we demonstrate this result in the case where $d_2(\eta_1,\eta_2)=\|\eta_1-\eta_2\|$, however, it holds for any metric $d_2$ by the equivalence of all metrics on $\mathcal{B}$. 

We control the approximate Bayesian computation posterior expectation of non-negative and bounded
functions $f_{T}({\theta }-{\theta }_{0})$ by
\begin{equation*}
\begin{split}
E_{\Pi _{\varepsilon }}\left\{ f_{T}({\theta }-{\theta }_{0})\right\} & =\int f_{T}({\theta }-{\theta }_{0})d\Pi _{\varepsilon }\left\{ {%
\theta } \mid \eta(y)\right\} \\
& =\int f_{T}({\theta }-{\theta }_{0})1\!\mathrm{l}_{\Vert {\theta }-{\theta 
}_{0}\Vert \leq \lambda _{T}}d\Pi _{\varepsilon }\left\{ {\theta } \mid {\eta(y) }\right\} +o_{P}(1) \\
& =\frac{\int_{\Vert {\theta }-{\theta }_{0}\Vert \leq \lambda _{T}}\pi({%
\theta })f_{T}({\theta }-{\theta }_{0})P_{{\theta }}\left\{ \Vert {\eta }({z}%
)-{\eta }({y})\Vert \leq \varepsilon _{T}\right\} d{\theta }}{\int_{\Vert {%
\theta }-{\theta }_{0}\Vert \leq \lambda _{T}}\pi({\theta })P_{{\theta }%
}\left\{ \Vert {\eta }({z})-{\eta }({y})\Vert \leq \varepsilon _{T}\right\} d{%
\theta }}+o_{P}(1),
\end{split}%
\end{equation*}%
where the second equality uses the posterior concentration of $\Vert {\theta 
}-{\theta }_{0}\Vert $ at the rate $\lambda _{T}{\gg }
1/v_{T}(1)$. For $b_0=b(\theta_0)$, define 
$$
Z_{T}^{0}={\Sigma }_{T}({\theta }_{0})\{{\eta(y) }-{b}%
_{0}\}, \quad Z_{T}={\Sigma }_{T}({\theta }_{0})\{{\eta }({z})-{b}({\theta })\},$$
with
\begin{equation*}
\begin{split}
{\Sigma }_{T}({\theta }_{0})\{{\eta }({z})-{\eta }({y})\}& ={\Sigma }_{T}({%
\theta }_{0})\{{\eta }({z})-{b}({\theta })\}+{\Sigma }_{T}({\theta }_{0})\{{b%
}({\theta })-{b}_{0}\}-{\Sigma }_{T}({\theta }_{0})\{{\eta }({y})-{b}_{0}\}
\\
& =Z_{T}+{\Sigma }_{T}(%
{\theta }_{0})\{{b}({\theta })-{b}_{0}\}-Z_{T}^{0}.
\end{split}%
\end{equation*}%
For fixed ${\theta }$, 
\begin{align*}
\Vert {\Sigma }_{T}^{-1}({\theta }_{0})&\left[ {\Sigma }_{T}({\theta }_{0})\{{%
\eta }({z})-{b}({\theta })\}-\Sigma_{T}({\theta }_{0})\{{b}({\theta })-{b}_{0}\}\right] \Vert \\
&\qquad \asymp \Vert {D}_{T}^{-1}\left[ {%
\Sigma }_{T}({\theta }_{0})\{{\eta }({z})-{b}({\theta })\}-\Sigma_{T}({\theta }_{0})\{{b}({\theta })-{b}_{0}\}\right] \Vert
\end{align*}%
and 
\begin{equation*}
{\Sigma }_{T}({\theta _{0}})\{{b}({\theta })-{b_{0}}\}-Z_{T}^{0}={\Sigma }%
_{T}({\theta _{0}})\nabla _{\theta }{b}({\theta _{0}})({\theta }-{\theta _{0}%
})\{1+o(1)\}-Z_{T}^{0}\in B.
\end{equation*}

\noindent {{Case (i)} :} We have $\lim_{T\to\infty}v_{T}(1)\varepsilon
_{T}=\infty $. Consider $x({\theta })=\varepsilon _{T}^{-1}\{{b}({\theta })-{%
b}_{0}\}$ and $f_{T}({\theta }-{\theta }_{0})=f\{\varepsilon _{T}^{-1}({%
\theta }-{\theta }_{0})\}$, where $f(\cdot)$ is a non-negative, continuous and
bounded function. On the event $\Omega _{n,0}(M)=\{\Vert Z_{T}^{0}\Vert \leq
M/2\}$, which has probability smaller than $\epsilon $ by choosing $M$ large
enough, we have that 
\begin{equation*}
P_{{\theta }}\left( \Vert Z_{T}-Z_{T}^{0}\Vert \leq M\right) \geq P_{{\theta 
}}\left( \Vert Z_{T}\Vert \leq M/2\right) \geq 1-\frac{c({\theta })}{%
M^{\kappa }}\geq 1-\frac{c_{0}}{M^{\kappa }}\geq 1-\epsilon 
\end{equation*}%
for all $\Vert {\theta }-{\theta }_{0}\Vert \leq \lambda _{T}$. Since, ${%
\eta }({z})-{\eta }({y})={\Sigma }_{T}^{-1}({\theta }_{0})(Z_{T}-Z_{T}^{0})+%
\varepsilon _{T}x(\theta)$, we have that on $\Omega _{n,0}$, 
\begin{equation*}
\begin{split}
P_{{\theta }}\left\{ \Vert {\Sigma }_{T}^{-1}({\theta }%
_{0})(Z_{T}-Z_{T}^{0})+\varepsilon _{T}x(\theta)\Vert \leq \varepsilon _{T}\right\} &
\geq P_{{\theta }}\left[ \Vert {\Sigma }_{T}^{-1}({\theta }%
_{0})(Z_{T}-Z_{T}^{0})\Vert \leq \varepsilon _{T}\{1-\Vert x(\theta)\Vert \}\right]  \\
& \geq P_{{\theta }}\left[ \Vert Z_{T}-Z_{T}^{0}\Vert \leq v
_{T}(1)\varepsilon _{T}\{1-\Vert x(\theta)\Vert \}\right] \geq 1-\epsilon 
\end{split}%
\end{equation*}%
as soon as $\Vert x(\theta)\Vert \leq 1-M/\{v_{T}(1)\varepsilon _{T}\}$ with $M$ as
above. This, combined with the continuity of $\pi(\cdot )$ at ${\theta }_{0}$
and assumption \ref{[A3"]}, implies that 
\begin{align}
& \int f\{\varepsilon _{T}^{-1}({\theta }-{\theta }_{0})\}d\Pi _{\varepsilon
}\left\{ {\theta } \mid \eta(y)\right\}  \nonumber\\
& \quad =\frac{\int_{\Vert {\theta }-{\theta }_{0}\Vert _{\leq }\lambda
_{T}}f\{\varepsilon _{T}^{-1}({\theta }-{\theta }_{0})\}P_{{\theta }}\left\{
\Vert {\Sigma }_{T}^{-1}({\theta }_{0})(Z_{T}-Z_{T}^{0})+\varepsilon
_{T}x(\theta)\Vert \leq \varepsilon _{T}\right\} d{\theta }}{\int_{\Vert {\theta }-{%
\theta }_{0}\Vert _{\leq }\lambda _{T}}P_{{\theta }}\left\{ \Vert {\Sigma }%
_{T}^{-1}({\theta }_{0})(Z_{T}-Z_{T}^{0})+\varepsilon _{T}x(\theta)\Vert \leq
\varepsilon _{T}\right\} d{\theta }}\{1+o(1)\}+o_{P}(1) \nonumber\\
& \quad =\frac{\int_{\Vert x({\theta })\Vert \leq 1-M/\{v_{T}(1)\varepsilon
_{T}\}}f\{\varepsilon _{T}^{-1}({\theta }-{\theta }_{0})\}d{\theta }}{%
\int_{\Vert x({\theta })\Vert \leq 1-M/\{v_{T}(1)\varepsilon _{T}\}}d{\theta }%
}\{1+o(1)\} \nonumber\\
&+\frac{\int_{\Vert {\theta }-{\theta }_{0}\Vert _{\leq }\lambda
_{T}}1\!\mathrm{l}_{\Vert x({\theta })\Vert >1-M/\{v_{T}(1)\varepsilon
_{T}\}}f\{\varepsilon _{T}^{-1}({\theta }-{\theta }_{0})\}P_{{\theta }}\left\{
\Vert {\Sigma }_{T}^{-1}({\theta }_{0})(Z_{T}-Z_{T}^{0})+\varepsilon
_{T}x(\theta)\Vert \leq \varepsilon _{T}\right\} d{\theta }}{\int_{\Vert x({\theta }%
)\Vert \leq 1-M/\{v_{T}(1)\varepsilon _{T}\}}d{\theta }} \nonumber\\
&\qquad\qquad 
\label{decompEf}
\end{align}
The first term is approximately equal to 
\begin{equation*}
N_{1}={\int_{\Vert {b}(\varepsilon _{T}u+{\theta }_{0})-{b}_{0}\Vert
\leq 1}f(u)du}\big/{\int_{\Vert {b}(\varepsilon _{T}u+{\theta }_{0})-{b}_{0}\Vert
\leq 1}du}
\end{equation*}%
and the regularity of the function ${\theta }\mapsto {b}({\theta })$
implies that 
\begin{equation*}
\int_{\Vert {b}(\varepsilon _{T}u+{\theta }_{0})-{b}_{0}\Vert \leq
\varepsilon _{T}}du=\int_{\Vert \nabla _{\theta }{b}({\theta _{0}})u\Vert
\leq 1}du+o(1)=\int_{u^{\intercal}B_{0}u\leq 1}du+o(1)
\end{equation*}%
with $B_{0}=\nabla _{\theta }{b}({\theta _{0}})^{\intercal}\nabla _{\theta }{b}({%
\theta _{0}})$. This leads to 
\begin{equation*}
N_{1}={\int_{u^{\intercal}B_{0}u\leq 1}f(u)du}\big/{\int_{u^{\intercal}B_{0}u\leq 1}du}.
\end{equation*}%
The second integral ratio in the right hand side of \eqref{decompEf}
converges to $0$. It can be split into an integral over $1+M/\{v
_{T}(1)\varepsilon _{T}\}\geq \Vert x({\theta })\Vert \geq 1-M/\{v
_{T}(1)\varepsilon _{T}\}$ and another over $1+M/\{v_{T}(1)\varepsilon
_{T}\}\leq \Vert x({\theta })\Vert $. The first part $N_2$ is bounded as
\begin{equation*}
N_{2}\leq \frac{\Vert f\Vert _{\infty }\int_{1+M/\{v_{T}(1)\varepsilon
_{T}\}\geq \Vert x({\theta })\Vert >1-M/\{v_{T}(1)\varepsilon _{T}\}}d{%
\theta }}{\int_{\Vert x({\theta })\Vert \leq 1-M/\{v_{T}(1)\varepsilon
_{T}\}}d{\theta }}\lesssim \{v_{T}(1)\varepsilon _{T}\}^{-1}=o(1)
\end{equation*}%
Since
\begin{equation*}
\begin{split}
& P_{{\theta }}\left\{ \Vert {\Sigma }_{T}^{-1}({\theta }%
_{0})(Z_{T}-Z_{T}^{0})+\varepsilon _{T}x(\theta)\Vert \leq \varepsilon _{T}\right\}
\leq P_{{\theta }}\left\{ \Vert {\Sigma }_{T}^{-1}({\theta }%
_{0})(Z_{T}-Z_{T}^{0})\Vert \geq \varepsilon _{T}\Vert x(\theta)\Vert -\varepsilon
_{T}\right\}  \\
& \qquad \leq P_{{\theta }}\left[ \Vert Z_{T}-Z_{T}^{0}\Vert \geq v
_{T}(1)\varepsilon _{T}\{\Vert x(\theta)\Vert -1\}\right] \leq {c_{0}}{[v
_{T}(1)\varepsilon _{T}\{\Vert x(\theta)\Vert -1\}]^{-\kappa }},
\end{split}%
\end{equation*}%
the second part of the second term, $N_3$, 
which is the integral over $\Vert x({\theta })\Vert >1+M/\{v _{T}(1)\varepsilon _{T}\}$, 
is bounded by
\begin{equation*}
\begin{split}
&\frac{\int_{\Vert {\theta }-{\theta }_{0}\Vert _{\leq }\lambda
_{T}}1\!\mathrm{l}_{\Vert x({\theta })\Vert >1+M/\{v_{T}(1)\varepsilon
_{T}\}}f\{\varepsilon _{T}^{-1}({\theta }-{\theta }_{0})\}P_{{\theta }}\left\{
\Vert {\Sigma }_{T}^{-1}({\theta }_{0})(Z_{T}-Z_{T}^{0})+\varepsilon
_{T}x(\theta)\Vert \leq \varepsilon _{T}\right\} d{\theta }}{\int_{\Vert x({\theta }%
)\Vert \leq 1-M/\{v_{T}(1)\varepsilon _{T}\}}d{\theta }} \\
& \lesssim M^{-\kappa }\varepsilon _{T}^{-k_{{\theta }}}\int_{2\geq \Vert
x(\theta)\Vert >1+M/\{v_{T}(1)\varepsilon _{T}\}}d{\theta }+2^{\kappa }\varepsilon
_{T}^{-k_{{\theta }}}\int_{2\leq \Vert x({\theta })\Vert }\{v
_{T}(1)\varepsilon _{T}\Vert x({\theta })\Vert \}^{-\kappa }d{\theta } \\
& \lesssim M^{-\kappa }+\varepsilon _{T}^{-k_{{\theta }}}\int_{c_{1}%
\varepsilon _{T}\leq \Vert {\theta }-{\theta }_{0}\Vert }\{v_{T}(1)\Vert
\nabla _{\theta }{b}({\theta }_{0})({\theta }-{\theta }_{0})\Vert
\}^{-\kappa }d{\theta }\lesssim M^{-\kappa },
\end{split}%
\end{equation*}%
provided $\kappa > 1$. Since $M$ can be chosen arbitrarily
large, putting $N_{1},N_{2}$ and $N_{3}$ together, we obtain that the approximate Bayesian computation
posterior distribution of $\varepsilon _{T}^{-1}({\theta }-{\theta }_{0})$
is asymptotically uniform over the ellipsoid $\{w:w^{\intercal}B_{0}w\leq 1\}$ and 
{(i)} is proved. \bigskip 

\noindent {{Case (ii)} :} We have $\infty >\lim_{T\to\infty}v
_{T}(1)\varepsilon _{T}=c>0$ and $\lim_{T\to\infty}v_{T}(k_{\eta })\varepsilon
_{T}=\infty $. We consider $$f_{T}({\theta }-{\theta }_{0})=\1_{{%
\Sigma }_{T}({\theta }_{0})\{{b}({\theta })-{b}_{0}\}-Z_{T}^{0}\in B}.$$ 
With an obvious abuse of notation, we let 
$$x={\Sigma }_{T}({\theta }_{0})\{{b}({\theta })-{b}_{0}\}-Z_{T}^{0}.$$

We choose $k_{0}$ such that, for all $j\leq k_{0}$, $\lim_{T\to\infty}v_{T}(j)\varepsilon
_{T}=c$ and for all $j>k_{0}$, $\lim_{T\to\infty}v_{T}(j)\varepsilon _{T}=\infty $. 
We write ${\Sigma }_{T}({\theta }_{0})={A}_{T}({\theta }_{0}){D}_{T}$, so
that ${A}_{T}({\theta }_{0})\rightarrow A({\theta }_{0})$ as $T\rightarrow\infty$, where ${A}({\theta }_{0})$ is positive definite and symmetric. Then,
\begin{equation*}
\begin{split}
P_{{\theta }}\left( \Vert {\Sigma }_{T}^{-1}({\theta }_{0})\left[ {\Sigma }%
_{T}({\theta }_{0})\{{\eta }({z})-{b}({\theta })\}-x\right] \Vert \leq
\varepsilon _{T}\right) & =P_{{\theta }}\left\{ \Vert {D}_{T}^{-1}{A}%
_{T}^{-1}({\theta }_{0})\left( Z_{T}-x\right) \Vert \leq \varepsilon
_{T}\right\} \\
& =P_{{\theta }}\left\{ \Vert {D}_{T}^{-1}(\tilde{Z}_{T}-x_{T})\Vert \leq
\varepsilon _{T}\right\},
\end{split}%
\end{equation*}%
where $\tilde{Z}_{T}={A}_{T}^{-1}({\theta }_{0})Z_{T}\longrightarrow \mathcal{N}%
\{0,{A}({\theta_0 })I_{k_{\eta }}{A}({\theta_0 })^{\intercal }\}$ and $x_{T}={A}%
_{T}^{-1}({\theta }_{0})x={A}^{-1}({\theta }_{0})x+o_{P}(1)$.

We then have for $M_{T}\rightarrow \infty $, such that $M_{T}\{v
_{T}(k_{0}+1)\varepsilon _{T}\}^{-2}=o(1)$, 
\begin{equation}
\begin{split}
P_{{\theta }}\big\{ &\Vert D_{T}^{-1}(\tilde{Z}_{T}-x_{T})\Vert \leq
\varepsilon _{T}\big\}  \leq P_{{\theta }}\left[ \sum_{j=1}^{k_{0}}\{\tilde{%
Z}_{T}(j)-x_{T}(j)\}^{2}\leq v_{T}(1)^{2}\varepsilon _{T}^{2}\right] \\
& \geq P_{{\theta }}\left( \sum_{j=1}^{k_{0}}\{\tilde{Z}_{T}(j)-x_{T}(j)\}^{2}%
\leq v_{T}(1)^{2}\varepsilon _{T}^{2}\left[1-M_{T}\{v
_{T}(k_{0}+1)\varepsilon _{T}\}^{-2}\right]\right) \\
& \quad -P_{{\theta }}\left[\sum_{j=k_{0}+1}^{k}\{\tilde{Z}%
_{T}(j)-x_{T}(j)\}^{2}>M_{T}^{-1}\{\varepsilon _{T}v
_{T}(k_{0}+1)\}^{-2}\right] \\
& \geq P_{{\theta }}\left( \sum_{j=1}^{k_{0}}\{\tilde{Z}_{T}(j)-x_{T}(j)\}^{2}%
\leq v_{T}(1)^{2}\varepsilon _{T}^{2} \left[1-M_{T}\{v
_{T}(k_{0}+1)\varepsilon _{T}\}^{-2}\right]\right) -o(1).
\end{split}
\label{decomp:inf}
\end{equation}%
This implies that, for all $x$ and all $\Vert {\theta }-{\theta }_{0}\Vert
\leq \lambda _{T}$ 
\begin{equation*}
P_{{\theta }}\left\{ \Vert {D}_{T}^{-1}(\tilde{Z}_{T}-x_{T})\Vert \leq
\varepsilon _{T}\right\} =P_{{\theta }}\left( \sum_{j=1}^{k_{0}}\left[
\left\{ {A}^{-1}({\theta }_{0})Z_{T}\right\} (j)-\{{A}^{-1}({\theta }%
_{0})x\}(j)\right] ^{2}\leq c\right) +o(1).
\end{equation*}%
Since ${A}^{-1}({\theta }_{0})x={D}_{T}{\nabla_{\theta}b}({\theta }_{0})({%
\theta }-{\theta }_{0})-{A}^{-1}({\theta }_{0})Z_{T}^{0}$, if 
$$
\text{  Leb}\left(
\sum_{j=1}^{k_{0}}\left[ \left\{ {\nabla_{\theta}b}({\theta }_{0})({\theta }-{%
\theta _{0}})\right\} _{[j]}\right]^{2}\leq c\varepsilon _{T}^{2}\right)
=\infty\,,
$$
then as in case {(i)} we can bound 
\begin{align*}
\Pi _{\varepsilon }&\left[ {\Sigma }_{T}({\theta }_{0})\{{b}(\theta)-{b}%
_{0}\}-Z_{T}^{0}\in B \mid \eta(y)\right] \\
&\quad \leq \frac{\int_{A^{-1}({\theta }%
_{0})x\in B}P_{{\theta }}\left( \sum_{j=1}^{k_{0}}\left[ \left\{
{A}^{-1}({\theta }_{0})Z_{T}\right\} (j)-z(j)\right] ^{2}\leq c\right) d{\theta }}{%
\int_{ \|{\theta } \| \leq M}P_{{\theta }}\left( \sum_{j=1}^{k_{0}}\left[ \left\{
{A}^{-1}({\theta }_{0})Z_{T}\right\} (j)-z(j)\right] ^{2}\leq c\right) d{%
\theta }}+o_{P}(1),
\end{align*}%
which goes to zero when $M$ goes to infinity. Since $M$ can be chosen
arbitrarily large, (12) is proven.

\bigskip

\noindent {{Case (iii)} :} We have $\lim_{T\to\infty}v_{T}(1)\varepsilon
_{T}=0$ and $\lim_{T\to\infty}v_{T}(k_{\eta })\varepsilon _{T}=\infty $. Again we
consider $$f_{T}({\theta }-{\theta }_{0})=\1_{{\Sigma }_{T}({\theta 
}_{0})\{{b}({\theta })-{b}_{0}\}-Z_{T}^{0}\in B}$$ and $x({\theta })={\Sigma }%
_{T}({\theta }_{0})\{{b}({\theta })-{b}_{0}\}-Z_{T}^{0}$. As in the
computations producing \eqref{decomp:inf}, and under assumption \ref{[A6"]}, we have
\begin{equation*}
\begin{split}
P_{{\theta }}& \left\{ \Vert {D}_{T}^{-1}(\tilde{Z}_{T}-x_{T})\Vert \leq
\varepsilon _{T}\right\}  \leq P_{{\theta }}\left[ \sum_{j=1}^{j_{\max }}v
_{T}(j)^{-2}\{\tilde{Z}_{T}(j)-x_{T}(j)\}^{2}\leq \varepsilon _{T}^{2}\right] \\
& \geq P_{{\theta }}\left[ \sum_{j=1}^{j_{\max }}v_{T}(j)^{-2}\{\tilde{Z}%
_{T}(j)-x_{T}(j)\}^{2}\leq \varepsilon _{T}^{2}/2\right]  - P_\theta \left[ \sum_{j \geq j_{\max } }\{\tilde{Z}_{T}(j)-x_{T}(j)\}^{2} > \epsilon_T^2 v_T(j_{\max }+1)^2/2\right]\\
& \geq \varphi _{j_{\max }}(x_{[k_{1}]})\{1+o(1)\}\prod_{j=1}^{j_{\max }}\{v
_{T}(j)\varepsilon _{T}\} -c_0 \{\epsilon_T v_T(j_{\max }+1)/2\}^{-\kappa}
\end{split}
\end{equation*}
uniformly when $\Vert {\theta }-{\theta }_{0}\Vert <\lambda _{T}$ where $\varphi
_{j_{\max }}$ is the zero mean Gaussian density in $j_{\max }$ dimensions,
with covariance $\{{A}({\theta }_{0})^{2}\}_{[j_{\max }]}$ . Since $\{\epsilon_T v_T(j_{\max }+1)/2\}^{-\kappa} = o \left[\prod_{j=1}^{j_{\max }}\{v
_{T}(j)\varepsilon _{T}\} \right]$ this
implies, as in case {(ii)}, that with probability going to one 
\begin{equation*}
\limsup_{T\to\infty}\Pi _{\varepsilon }\left[ {\Sigma }_{T}({\theta }_{0})\{{b}(\theta)-{b}%
_{0}\}-Z_{T}^{0}\in B \mid \eta(y)\right] \lesssim \frac{\int_{{A}({\theta }%
_{0})B}\varphi _{j_{\max }}(x_{[j_{\max }]})dx}{\int_{ \| x \| \leq M}\varphi
_{j_{\max }}(x_{[j_{\max }]})dx}\lesssim  M^{-(k_{\eta }-j_{\max })}
\end{equation*}%
and choosing $M$ arbitrary large leads to equation (11) in the text.\bigskip

\noindent {{Case (iv)} :} We have $\lim_{T\to\infty}v
_{T}(j)\varepsilon _{T}=c>0$ for all $j\leq k_{\eta }$. To prove equation
(13) in the text, we use the computation of case {(ii)} with $%
k_{0}=k_{\eta }$, so that \eqref{decomp:inf} implies that for all $x$
\begin{equation*}
\begin{split}
P_{{\theta }}\left\{ \Vert {D}_{T}^{-1}(\tilde{Z}_{T}-x)\Vert \leq
\varepsilon _{T}\right\} & =P_{{\theta }}\left\{ \Vert \tilde{Z}%
_{T}-x\Vert ^{2}\leq v_{T}(1)^{2}\varepsilon _{T}^{2}\right\} \\
& = P\left\{ \Vert {A}^{-1}({\theta }_{0})  Z_\infty -{A}^{-1}({\theta }_{0})x\Vert
^{2}\leq c^{2}\right\} +o(1)\\
&
\end{split}%
\end{equation*}%
uniformly in $\| \theta - \theta_0\|\leq \delta $ where $Z_\infty \sim \mathcal N(0, I_{k_\eta})$.  

We set $u = \{ \nabla_\theta b(\theta_0)^\intercal \Sigma_T(\theta_0 )^\intercal
\Sigma_T(\theta_0) \nabla b(\theta_0)\}^{1/2}(\theta- \theta_0) $. When
$\|\theta -\theta^*\|\leq \lambda_T$, $x(\theta) = \Sigma_T(\theta_0) \nabla
b(\theta_0) (\theta - \theta_0) ( 1 + o(1)) = x(u) \{1+ o(1)\} $, we can write
$x(\theta) = x(u) $ and $\|x(u)\|= \|u\|$. 
Choosing $M$ and $T$ large enough, by the dominated convergence theorem,
\begin{equation*}
\begin{split}
\Pi _{\varepsilon }\left[ {\Sigma }_{T}({\theta }_{0})\{{b}(\theta)-{b}%
_{0}\}\right.&-\left.Z_{T}^{0}\in B \mid \eta(y)\right] \leq \frac{\int_{x(u)\in B}P\left\{ \Vert {A}^{-1}({\theta }_{0})Z_{\infty}-{A}^{-1}({\theta }_{0})x(u) \Vert
^{2}\leq c^{2}\right\} du+o_p(1) }{\int_{ \| u \| \leq M}P\left\{ \Vert {A}^{-1}(%
{\theta }_{0})Z_{\infty}-{A}^{-1}({\theta }_{0})x(u)\Vert ^{2}\leq c^{2}\right\} du +o_p(1)}
\\
& \geq \frac{\int_{x(u) \in B}P\left\{ \Vert {A}^{-1}({\theta }%
_{0})Z_{\infty}-{A}^{-1}({\theta }_{0})x(u)\Vert ^{2}\leq c^{2}\right\} du+ o_p(1)}{%
\int_{ \| u \| \leq M}P\left\{ \Vert {A}^{-1}({\theta }_{0})Z_{\infty}-{A}%
^{-1}({\theta }_{0})x(u) \Vert ^{2}\leq c^{2}\right\} du +o_p(1) }
\end{split}%
\end{equation*}%
Since $M$ can be chosen arbitrarily large and since, when $M$ goes to infinity, we have
\begin{equation*}
\int_{ \| u \| \leq M}P\left\{ \Vert \tilde{Z}_{\infty}-{A}^{-1}({\theta }%
_{0})x(u) \Vert ^{2}\leq c^{2}\right\} du \rightarrow \int_{u\in \mathbb{R}%
^{k_{\theta }}}P\left\{ \Vert \tilde{Z}_{\infty}-{A}^{-1}({\theta }%
_{0})x(u) \Vert ^{2}\leq c^{2}\right\} du <\infty ,
\end{equation*}%
the result follows. \medskip

\noindent {{Case (v)} :} We have $\lim_{T\to\infty}v_{T}(k)\varepsilon
_{T}=0.$ Take ${\Sigma }_{T}({\theta }_{0})={A}_{T}({\theta }_{0}){D}_{T}$.
For some $\delta >0$ and all $\Vert {\theta }-{\theta }_{0}\Vert \leq \delta 
$, 
\begin{flalign*}
P_{\theta}\left[\left\|{D}_{T}^{-1}\{{A}_{T}^{-1}({\theta}_{0}){Z}_{T}-{A}_{T}^{-1}({\theta}_{0})x\}\right\|\leq \varepsilon_{T}\right]&=
P_{{\theta}}\left[\left\{{A}_{T}^{-1}({\theta}_{0}){Z}_{T}-{A}_{T}^{-1}({\theta}_{0})x\right\}\in B_{T}\right]+o(1). 
\end{flalign*}From both assertions of assumption \ref{[A6"]} and by the dominated
convergence theorem, the above implies for $j_{\max }=k_{\eta }$ that
\begin{flalign*}
\frac{1}{\prod_{j=1}^{k_{\eta}}\varepsilon_{T}v_{T}(j) }
\int P_{{\theta}}\left[\left\{{A}_{}^{-1}({\theta}_{0}){Z}_{T}-{A}_{}^{-1}({\theta}_{0})x\right\}\in B_{T}\right]dx=\int_{}\varphi_{k_{\eta}}(x)dx+o(1)=1+o(1)\,.
\end{flalign*}Likewise, similar arguments yield 
\begin{flalign*}
\frac{1}{\prod_{j=1}^{k_{\eta}}\varepsilon_{T}v_{T}(j) }\int_{}\1_{x\in B}P_{{\theta}}\left[\left\{{A}_{}^{-1}({\theta}_{0}){Z}_{T}-{A}_{}^{-1}({\theta}_{0})x\right\}\in B_{T}\right]dx&=\int_{}\1_{x\in B}\varphi_{k_{\eta}}(x)dx+o(1)\\&=\Phi_{k_{\eta}}(B)+o(1). 
\end{flalign*}Together, these two equivalences yield the result in case (v).
\end{proof}

\subsection{Proof of Theorem 3}


\noindent {{Case (i)} :}
Define $b=b(\theta )$ and $b_{0}=b(\theta _{0})$ and, with a slight abuse of notation, in this proof we let $Z_{T}^{0}=v_{T}\{\eta ({y})-b_0\}$ and $x=v _{T}(b-b_{0})-Z_{T}^{0}.$  We approximate the ratio 
\begin{equation*}
E_{\Pi _{\varepsilon }}\{v_{T}({b}-b_{0})\}-Z_{T}^{0}=\frac{N_{T}}{D_{T}}=
\frac{\int xP_{x}\left\{  \mid \eta ({z})-\eta ({y}) \mid \leq \varepsilon _{T}\right\}
\pi\{b_{0}+(x+Z_{T}^{0})/v_{T}\}dx}{\int P_{{x}}\left\{  \mid \eta ({z})-\eta ({y}
) \mid \leq \varepsilon _{T}\right\} \pi\{b_{0}+(x+Z_{T}^{0})/v_{T}\}dx}
\end{equation*}%
We first approximate the numerator $N_{T}$: $v_{T}\{\eta ({z})-\eta ({y}%
)\}=v_{T}(\eta ({z})-b)+x$ and $b=b_{0}+(x+Z_{T}^{0})/v_{T}$.
Denote $Z_{T}=v_{T}\{\eta ({z})-b\}$, then 
\begin{equation}
\begin{split}
N_{T}& =\int xP_{{x}}\left\{  \mid \eta ({z})-\eta ({y}) \mid \leq \varepsilon
_{T}\right\} \pi\{b_{0}+(x+Z_{T}^{0})/v_{T}\}dx \\
& =\int_{ \mid x \mid \leq v_{T}\varepsilon _{T}-M}xP_{{x}}\left(  \mid Z_{T}+x \mid \leq v
_{T}\varepsilon _{T}\right) \pi\{b_{0}+(x+Z_{T}^{0})/v_{T}\}dx \\
& +\int_{ \mid x \mid \geq v_{T}\varepsilon _{T}-M}xP_{{x}}\left( \mid Z_{T}+x \mid \leq v
_{T}\varepsilon _{T}\right) \pi\{b_{0}+(x+Z_{T}^{0})/v_{T}\}dx,
\end{split}
\label{decomp:NT}
\end{equation}%
where the condition $\lim_{T}v_{T}\varepsilon _{T}=\infty $ is used in
the representation of the real line over which the integral defining $N_{T}$
is specified.

{We start by studying the first integral term in \eqref{decomp:NT}.} If $%
0\leq x\leq v_{T}\varepsilon _{T}-M$, then 
\begin{equation*}
\begin{split}
1\geq P_{{x}}\left(  \mid Z_{T}+x \mid \leq v_{T}\varepsilon _{T}\right) & =1-P_{{x}%
}\left( Z_{T}>v_{T}\varepsilon _{T}-x\right) -P_{{x}}\left( Z_{T}<-v
_{T}\varepsilon _{T}-x\right) \\
& \geq 1-2(v_{T}\varepsilon _{T}-x)^{-\kappa }.
\end{split}%
\end{equation*}%
Using a similar argument for $x\leq 0$, we obtain, for all $ \mid x \mid \leq v
_{T}\varepsilon _{T}-M$, 
\begin{equation*}
1-2(v_{T}\varepsilon _{T}- \mid x \mid )^{-\kappa }\leq P_{{x}}\left(  \mid Z_{T}+x \mid \leq
v_{T}\varepsilon _{T}\right) \leq 1
\end{equation*}%
and choosing $M$ large enough implies that if $\kappa >2$, 
\begin{equation*}
\begin{split}
N_{1}& =\int_{ \mid x \mid \leq v_{T}\varepsilon _{T}-M}x P_{{x}}\left(
 \mid Z_{T}+x \mid \leq v_{T}\varepsilon _{T}\right) \pi\{b_{0}+(x+Z_{T}^{0})/v_{T}\}dx \\
& =\int_{ \mid x \mid \leq v_{T}\varepsilon _{T}-M}x\pi\{b_{0}+(x+Z_{T}^{0})/v_{T}\}dx+O(M^{-\kappa +2})\,.
\end{split}%
\end{equation*}%
A Taylor expansion of $\pi\{b_{0}+(x+Z_{T}^{0})/v_{T}\}$ around $\gamma
_{0}=b_{0}+Z_{T}^{0}/v_{T}$ then leads to, for $\nabla_{b}^{j}\pi(\theta)$ 
denoting the $j$-th derivative of $\pi(b)$ with respect to $b$,  
\begin{equation*}
\begin{split}
N_{1}& =2\sum_{j=1}^{k}\frac{\nabla_{b}^{(2j-1)}\pi(\gamma _{0})}{(2j-1)!(2j+1)v
_{T}^{2j-1}}(\varepsilon _{T}v_{T})^{2j+1}+O(M^{-\kappa
+2})+O(\varepsilon _{T}^{2+\beta }v_{T}^{2})+o_{P}(1) \\
& =2v_{T}^{2}\sum_{j=1}^{k}\frac{\nabla_{b}^{(2j-1)}\pi(\gamma _{0})}{(2j-1)!(2j+1)}%
\varepsilon _{T}^{2j+1}+O(M^{-\kappa +2})+O(\varepsilon
_{T}^{2+\beta }v_{T}^{2})+o_{P}(1),
\end{split}%
\end{equation*}%
where $k=\lfloor \beta /2\rfloor .$ We split the second
integral of \eqref{decomp:NT} over $v_{T}\varepsilon _{T}-M\leq  \mid x \mid \leq
v_{T}\varepsilon _{T}+M$ and over $ \mid x \mid \geq v_{T}\varepsilon _{T}+M$. We
treat the latter as before: with probability going to one, 
\begin{equation*}
\begin{split}
 \mid N_{3} \mid & \leq \int_{ \mid x \mid \geq v_{T}\varepsilon _{T}+M} \mid x \mid P_{{x}}\left(
 \mid Z_{T}+x \mid \leq v_{T}\varepsilon _{T}\right) \pi\{b_{0}+(x+Z_{T}^{0})/v
_{T}\}dx \\
& \leq \int_{ \mid x \mid \geq v_{T}\varepsilon _{T}+M}\frac{ \mid x \mid c\{b_{0}+%
(x+Z_{T}^{0})/v_{T}\}}{( \mid x \mid -v_{T}\varepsilon _{T})^{\kappa }}%
\pi\{b_{0}+(x+Z_{T}^{0})/v_{T}\}dx \\
& \leq c_{0}\sup_{|x|\geq v_{T}\varepsilon_{T}}| \pi(x)|\int_{v_{T}\varepsilon _{T}+M\leq
 \mid x \mid \leq \delta v_{T}}\frac{ \mid x \mid }{( \mid x \mid -v_{T}\varepsilon _{T})^{\kappa }}dx+%
\frac{v_{T}}{(\delta v_{T})^{\kappa -1}}\int c(\theta )d\Pi (\theta)\\
& \lesssim M^{-\kappa +2}+O(v_{T}^{-\kappa +2})\,.
\end{split}%
\end{equation*}%
Finally, we study the second integral term for $N_{T}$ in \eqref{decomp:NT}
over $v_{T}\varepsilon _{T}-M\leq  \mid x \mid \leq v_{T}\varepsilon _{T}+M$.
Using the assumption that $\pi(\cdot )$ is H\"{o}lder we obtain that 
\begin{equation*}
\begin{split}
 \mid N_{2} \mid & =\left\vert \int_{v_{T}\varepsilon _{T}-M}^{v_{T}\varepsilon
_{T}+M}xP_{{x}}\left(  \mid Z_{T}+x \mid \leq v_{T}\varepsilon _{T}\right)
\pi\{b_{0}+(x+Z_{T}^{0})/v_{T}\}dx\right. \\
& \quad \left. +\int_{-v_{T}\varepsilon _{T}-M}^{-v_{T}\varepsilon
_{T}+M}xP_{{x}}\left(  \mid Z_{T}+x \mid \leq v_{T}\varepsilon _{T}\right)
\pi\{b_{0}+(x+Z_{T}^{0})/v_{T}\} dx\right\vert \\
& \leq \pi(b_{0})\left\vert \int_{v_{T}\varepsilon _{T}-M}^{v
_{T}\varepsilon _{T}+M}xP_{{x}}\left(  \mid Z_{T}+x \mid \leq v_{T}\varepsilon
_{T}\right) dx+\int_{-v_{T}\varepsilon _{T}-M}^{-v_{T}\varepsilon
_{T}+M}xP_{{x}}\left(  \mid Z_{T}+x \mid \leq v_{T}\varepsilon _{T}\right)
dx\right\vert \\
& \qquad +L\cdot M\varepsilon _{T}^{1+\beta \wedge 1}v_{T}^{\beta \wedge
1}+o_{P}(1) \\
& \lesssim \left\vert v_{T}\varepsilon _{T}\int_{-M}^{M}\left\{
P_{y}\left( Z_{T}\leq -y\right) -P_{y}\left( Z_{T}\geq -y\right) \right\}
dy\right\vert \\
& +\left\vert v_{T}\varepsilon _{T}\int_{-M}^{M}y\left\{ P_{y}\left(
Z_{T}\leq -y\right) +P_{y}\left( Z_{T}\geq -y\right) \right\} dy\right\vert
+O(M\varepsilon _{T}^{1+\beta \wedge 1}v_{T}^{\beta \wedge 1})+o_{P}(1),
\end{split}%
\end{equation*}%
with $M$ fixed but arbitrarily large. By the dominated convergence theorem
and the locally uniform Gaussian limit of $Z_{T}$,{\ for any arbitrarily large, but fixed $M$%
, } 
\begin{equation*}
\int_{-M}^{M}\left\{ P_{y}\left( Z_{T}\leq -y\right) -P_{y}\left( Z_{T}\geq
-y\right) \right\} dy=Mo(1)
\end{equation*}%
and 
\begin{equation*}
\int_{-M}^{M}y\left\{ P_{y}\left( Z_{T}\leq -y\right) +P_{y}\left( Z_{T}\geq
-y\right) \right\} dy=\int_{-M}^{M}y\left\{ 1+o(1)\right\} dy=M^{2}o(1).
\end{equation*}%
This implies that 
\begin{equation*}
N_{2}\lesssim M^{2}o(v_{T}\varepsilon _{T})+M\varepsilon _{T}^{1+\beta
\wedge 1}v_{T}^{\beta \wedge 1}+o_{P}(1).
\end{equation*}%
Therefore, regrouping all terms, and since $\varepsilon _{T}^{1+\beta \wedge
1}v_{T}^{\beta \wedge 1}=o(v_{T}\varepsilon _{T})$ for all $\beta >0$ and
$\varepsilon _{T}=o(1)$, we obtain the representation 
\begin{equation*}
N_{T}=2v_{T}^{2}\sum_{j=1}^{k}\frac{\nabla_{b}^{(2j-1)}\pi(\gamma _{0})}{(2j-1)!(2j+1)%
}\varepsilon _{T}^{2j+1}+M^{2}o(v_{T}\varepsilon _{T})+O(M^{-\kappa
+2})+O(v_{T}^{-\kappa +2})+O(\varepsilon _{T}^{2+\beta }v
_{T}^{2})+o_{P}(1)\,.
\end{equation*}
We now study the denominator in a similar manner. This leads to 
\begin{equation*}
\begin{split}
D_{T}& =\int P_{{x}}\left\{  \mid \eta ({z})-\eta ({y}) \mid \leq \varepsilon
_{T}\right\} \pi\{b_{0}+(x+Z_{T}^{0})/v_{T}\}dx \\
& =\int_{ \mid x \mid \leq v_{T}\varepsilon _{T}-M}\pi\{b_{0}+(x+Z_{T}^{0})/v
_{T}\}\{1+o(1)\}dx+O(1) \\
& =2\pi(b_{0})v_{T}\varepsilon _{T} \{1+o_{P}(1)\}.
\end{split}%
\end{equation*}%
Combining $D_{T}$ and $N_{T}$, we obtain, 
$\varepsilon _{T}=o(1)$, 
\begin{equation}
\frac{N_{T}}{D_{T}}=v_{T}\sum_{j=1}^{k}\frac{\nabla_{b}^{(2j-1)}\pi(b_{0})}{%
\pi(b_{0})(2j-1)!(2j+1)}\varepsilon _{T}^{2j}+o_{P}(1)+O(\varepsilon
_{T}^{1+\beta }v_{T}) . \label{postmean}
\end{equation}%
Using the definition of $N_{T}/D_{T}$, dividing \eqref{postmean} by $v
_{T} $, and rearranging terms yields 
\begin{equation*}
E_{\Pi_{\varepsilon}}\left( b-b_{0}\right) =\frac{Z_{T}^{0}}{v_{T}}%
+\sum_{j=1}^{k}\frac{\nabla_{b}^{(2j-1)}\pi(b_{0})}{\pi(b_{0})(2j-1)!(2j+1)}\varepsilon
_{T}^{2j}+O(\varepsilon _{T}^{1+\beta })+o_{P}(1/v_{T}),
\end{equation*}
To obtain the posterior mean of $\theta $, we write 
\begin{equation*}
\theta =b^{-1}\{b(\theta )\}=\theta _{0}+\sum_{j=1}^{\lfloor \beta \rfloor }%
\frac{\{b(\theta )-b_{0}\}^{j}}{j!}\nabla_{b}^{(j)}b^{-1}(b_{0})+R(\theta ),
\end{equation*}%
where $ \mid R(\theta ) \mid \leq L \mid b(\theta )-b_{0} \mid ^{\beta }$ provided $ \mid b(\theta
)-b_{0} \mid \leq \delta $. We compute the approximate Bayesian mean of $\theta $ by splitting the
range of integration into $ \mid b(\theta )-b_{0} \mid \leq \delta $ and $ \mid b(\theta
)-b_{0} \mid >\delta $. A Cauchy-Schwarz inequality leads to 
\begin{equation*}
\begin{split}
E_{\Pi _{\varepsilon }}\left\{  \mid \theta -\theta _{0} \mid \right.&\left.1\!\mathrm{l}_{ \mid b(\theta
)-b_{0} \mid >\delta }\right\}\\ 
& =\frac{1}{2\varepsilon _{T}v
_{T}\pi(b_{0})\{1+o_{P}(1)\}}\int_{ \mid b(\theta )-b_{0} \mid >\delta } \mid \theta -\theta
_{0} \mid P_{\theta }\left\{  \mid \eta ({z})-\eta ({y}) \mid \leq \varepsilon _{T}\right\}
\pi(\theta )d\theta  \\
& \leq {2^{\kappa }v_{T}^{-\kappa }\delta ^{-\kappa }\left\{ \int_{\Theta }(\theta -\theta
_{0})^{2}\pi(\theta )d\theta \right\} ^{1/2}\left\{ \int_{\Theta }c(\theta
)^{2}\pi(\theta )d\theta \right\} ^{1/2}} \{1+o_{P}(1)\} \\
& =o_{P}(1/v_{T}),
\end{split}%
\end{equation*}%
provided $\kappa >1$. To control the term over $|b(\theta)-b_0|\leq\delta$, we use computations
similar to earlier ones so that 
\begin{equation*}
E_{\Pi _{\varepsilon }}\left\{ (\theta -\theta _{0})1\!\mathrm{l}%
_{ \mid b(\theta )-b_{0} \mid \leq \delta }\right\} =\sum_{j=1}^{\lfloor \beta \rfloor }%
\frac{\nabla_{b}^{(j)}b^{-1}(b_{0})}{j!}E_{\Pi _{\varepsilon }}\left[ \{b(\theta
)-b_{0}\}^{j}\right] +o_{P}(1/v_{T}),
\end{equation*}%
where, for $j\geq 2$ and $\kappa >j+1$, 
\begin{equation*}
\begin{split}
E_{\Pi _{\varepsilon }}\left[ \{b(\theta )-b_{0}\}^{j}\right] & =\frac{1}{%
v_{T}^{j}}\frac{\int_{ \mid x \mid \leq \varepsilon _{T}v
_{T}-M}x^{j}\pi\{b_{0}+(x+Z_{T}^{0})/v_{T}\}dx}{2\varepsilon _{T}v
_{T}\pi(b_{0})}+o_{P}(1/v_{T}) \\
& =\sum_{l=0}^{k}\frac{\nabla_{b}^{(l)}\pi(b_{0})}{2\varepsilon _{T}v
_{T}^{j+l+1}\pi(b_{0})l!}\int_{ \mid x \mid \leq \varepsilon _{T}v
_{T}-M}x^{j+l}dx+o_{P}(1/v_{T})+O(\varepsilon _{T}^{1+\beta }) \\
& =\sum_{l=\lceil j/2\rceil }^{\lfloor (j+k)/2\rfloor }\frac{\varepsilon
_{T}^{2l}\nabla_{b}^{(2l-j)}\pi^{}(b_{0})}{\pi(b_{0})(2l-j)!}+o_{P}(1/v_{T})+O(\varepsilon
_{T}^{1+\beta }).
\end{split}%
\end{equation*}%
This implies, in particular, that 
\begin{equation*}
E_{\Pi _{\varepsilon}}( \theta -\theta _{0}) =\frac{
Z_{T}^{0}\{\nabla_{b}b^{-1}(b_{0})\}}{v_{T}}+\sum_{j=1}^{\lfloor \beta \rfloor }
\frac{\nabla_{b}^{(j)}b^{-1}(b_{0})}{j!}\sum_{l=\lceil j/2\rceil }^{\lfloor
(j+k)/2\rfloor }\frac{\varepsilon _{T}^{2l}\nabla_{b}^{(2l-j)}\pi(b_{0})}{\pi(b_{0})(2l-j)!
}+o_{P}(1/v_{T})+O(\varepsilon_{T}^{1+\beta}).
\end{equation*}
Hence, if $\varepsilon _{T}^{2}=o(1/v_{T})$ and $\beta \geq 1$, 
\begin{equation*}
E_{\Pi _{\varepsilon }}\left( \theta -\theta _{0}\right)= \{\nabla_{\theta}b(\theta _{0})\}^{-1}{Z_{T}^{0}}/{v_{T}}+o_{P}(1/v_{T})
\end{equation*}%
and $E_{\Pi _{\varepsilon }}\left\{ v_{T}(\theta -\theta _{0})\right\}
\longrightarrow \mathcal{N}[0,V(\theta_{0})/\{\nabla_{\theta}b(\theta _{0})\}^{2}]$, while if $%
v_{T}\varepsilon _{T}^{2}\rightarrow \infty $ 
\begin{equation*}
E_{\Pi _{\varepsilon }}\left( \theta -\theta _{0}\right) =\varepsilon
_{T}^{2}\left[ \frac{\nabla_{b}\pi(b_{0})}{3\pi(b_{0})\nabla_{\theta}b(\theta _{0})}-%
\frac{\nabla_{\theta}^{(2)}b(\theta _{0})}{2\{\nabla_{\theta}b(\theta _{0})\}^{2}}\right]
+O(\varepsilon _{T}^{4})+o_{P}(1/v_{T}),
\end{equation*}%
assuming $\beta \geq 3$.
\medskip 

\noindent {{Case (ii)} :} Recall
that $b=b(\theta )$, $b_0=b(\theta_0)$, and define 
\begin{equation*}
E_{\Pi _{\varepsilon }}\left( b\right) =\int \frac{bP_{b}\left\{  \mid \eta ({y}%
)-\eta ({z}) \mid \leq \varepsilon _{T}\right\} \pi(b)db}{\int P_{b}\left\{  \mid \eta ({y%
})-\eta ({z}) \mid \leq \varepsilon _{T}\right\} \pi(b)db}.
\end{equation*}%
Considering the change of variables $b\mapsto x=v_{T}(b-b_{0})-Z_{T}^{0}$
and using the above equation we have 
\begin{equation*}
E_{\Pi _{\varepsilon }}\left( b\right) =\int \frac{( b_{0}+(x+Z_{T}^{0})/v
_{T})P_{x}\left\{  \mid \eta ({y})-\eta ({z}) \mid \leq \varepsilon _{T}\right\}
\pi\{b_{0}+(x+Z_{T}^{0})/v_{T}\}dx}{\int P_{x}\left\{  \mid \eta ({y})-\eta ({z}%
) \mid \leq \varepsilon _{T}\right\} \pi\{b_{0}+(x+Z_{T}^{0})/v_{T}\}dx},
\end{equation*}%
which can be rewritten as 
\begin{flalign*}
E_{\Pi _{\varepsilon }}\left\{v_{T}({b}-b_{0})\right\}-Z^{0}_{T}=\int \frac{x
P_{x}\left\{ \mid \eta({y})-\eta({z}) \mid \leq\varepsilon_{T}\right\}\pi\{ b_0 +
(x+Z^{0}_{T})/v_{T}\}dx}{\int
P_{x}\left\{ \mid \eta({y})-\eta({z}) \mid \leq\varepsilon_{T}\right\}\pi\{ b_0 +
(x+Z^{0}_{T})/v_{T}\}dx}.
\end{flalign*}
Recalling that $v_{T}\{\eta ({z})-\eta ({y})\}=v_{T}\{\eta (%
{z})-b\}+v_{T}(b-b_{0})-Z_{T}^{0}=Z_{T}+x$ we have 
\begin{flalign*}
E_{\Pi _{\varepsilon }}\left\{v_{T}(b-b_{0})\right\}-Z^{0}_{T}=\int \frac{x P_{x}\left( \mid Z_{T}+x \mid \leq v_{T}\varepsilon_{T}\right)\pi\{b_{0}+(x+Z_{T}^{0})/v_{T}\} dx}{\int P_{x}\left( \mid Z_{T}+x \mid \leq v_{T}\varepsilon_{T}\right)\pi
\{ b_0+ (x+Z^{0}_{T})/v_{T}\} dx}=\frac{N_{T}}{D_{T}}.
\end{flalign*}By injectivity of the map $\theta \mapsto b(\theta )$
in assumption \ref{[A3]} and assumption \ref{[A4]}, the result follows
when $E_{\Pi _{\varepsilon }}\left\{ v_{T}({b}-b_{0})\right\} -Z_{T}^{0}=o_{P}(1)$.

Consider first the denominator. Define $h_{T}=v_{T}\varepsilon _{T}$ and $V_0=V(\theta_0)=\lim_{T\to\infty}\text{var}[v_{T}\{\eta(y)-b_0\}]$. Using arguments that mirror those in the proof of Theorem 2 part (v), by
assumption \ref{[A6"]} and the dominated convergence theorem 
\begin{flalign*}
\frac{D_{T}}{\pi(b_{0})h_{T}}=h_T^{-1}\int {P_{x}( \mid Z_{T}+x \mid \leq h_{T})}dx+o_P(1)=\int
\varphi(x/V_{0}^{1/2}) dx+o_P(1)=1+o_P(1),
\end{flalign*}where the second equality follows from assumption \ref{[A6]}
and the dominated convergence theorem. The result follows if $N_{T}/h_{T}=o_{P}(1)$. To this
end, define $P_{x}^{\ast }( \mid Z_{T}+x \mid \leq h_{T})=P_{x}( \mid Z_{T}+x \mid \leq
h_{T})/h_{T}$ and, if $h_{T}=o(1)$ by assumptions \ref{[A6]} and \ref{[A7]}, 
\begin{flalign*}
\frac{N_{T}}{h_{T}}&= \int {x P_{x}^{*}( \mid Z_{T}+x \mid \leq h_{T})\pi\{b_{0}+(x+Z_{T}^{0})/v_{T}\}} dx\\
&=\pi(b_{0})\int x 
\varphi(x/V_{0}^{1/2}) dx+\int x\left\{ P_{x}^{*}( \mid Z_{T}+x \mid \leq h_{T})-\varphi(x/V_{0}^{1/2})\right\}\\
&\qquad \times \pi\{b_{0}+(x+Z_{T}^{0})/v_{T}\} dx +o_P(1).
\end{flalign*}
If $h_{T}\rightarrow c>0$, then 
\begin{align}
\frac{N_{T}}{h_{T}}& =\pi(b_{0})\int x\cdot\text{pr}\left\{ \mid \mathcal{N}(0,1)+x/V_{0}^{1/2} \mid \leq c/V_{0}^{1/2}\right\} dx \nonumber\\
& \quad +\int x\left[ P_{x}^{\ast }( \mid Z_{T}+x \mid \leq h_{T})-\text{pr}\{ \mid \mathcal{N}%
(0,1)+x/V_{0}^{1/2} \mid \leq c/V_{0}^{1/2}\}\right]\pi\{b_{0}+(x+Z_{T}^{0})/v_{T}\} dx \nonumber\\
&\qquad +o_{P}(1).
\label{bis}
\end{align}
The result follows if 
$$\int x\left\{ P_{x}^{\ast }( \mid Z_{T}+x \mid \leq
h_{T})-\varphi ( x/V_{0}^{1/2})\right\} \pi\{b_{0}+(x+Z_{T}^{0})/v
_{T}\}dx=o_{P}(1),$$ 
respectively, $P_{x}^{\ast }( \mid Z_{T}+x \mid \leq h_{T})-\text{pr}\{ \mid \mathcal{N}%
(0,1)+x/V_{0}^{1/2} \mid \leq c/V_{0}^{1/2}\}=o(1)$, for which a
sufficient condition is that 
\begin{equation}
\int  \mid x \mid \left\vert P_{x}^{\ast }( \mid Z_{T}+x \mid \leq h_{T})-\varphi
(x/V_{0}^{1/2}) \right\vert \pi\{b_{0}+(x+Z_{T}^{0})/v_{T}\}dx=o_{P}(1),
\label{show}
\end{equation}%
or the equivalent in the case $h_{T}\rightarrow c>0$.

To show that the integral in \eqref{show} is $o_{P}(1)$ we break the region
of integration into three areas: (i) $ \mid x \mid \leq M$; (ii) $M\leq  \mid x \mid \leq \delta
v_{T}$; (iii) $ \mid x \mid \geq \delta v_{T}$.

\medskip

\noindent{Area
(i):} When $ \mid x \mid \leq M$, the following equivalences are satisfied: 
\begin{flalign*}
\sup_{x:  \mid x \mid \leq
M}& \mid \pi\{b_{0}+(x+Z_{T}^{0})/v_{T}\}-\pi(b_{0}) \mid =o_{P}(1),\\
\sup_{  \mid \theta  -\theta^* \mid \leq
1/v_T}& \mid P^{*}_{\theta}( \mid Z_{T}+x \mid \leq h_{T})-\varphi( x/V_{0}^{1/2}) \mid =o_{P}(1).
\end{flalign*}The first equation is satisfied by assumption \ref{[A7]} and the fact
that by assumption \ref{[A4]} $Z_{T}^{0}/v_{T}=o_{P}(1)$. The second term follows
from assumption \ref{[A7]} . We can now conclude that equation \eqref{show}
is $o_{P}(1)$ over $ \mid x \mid \leq M$, using the dominated convergence theorem.

The same holds for the first term in equation \eqref{bis}, without requiring assumption \ref{[A7]}.

\medskip

\noindent{Area (ii):} When $M\leq |x| \leq \delta v_{T}$, the integral of the
second term is finite and can be made arbitrarily small for $M$
large enough. Therefore, it suffices to show that 
\begin{equation*}
\int_{M\leq |x| \leq \delta v_{T}}|x| P_{x}^{\ast }(\mid Z_{T}+x\mid \leq
h_{T})\pi\{b_{0}+(x+Z_{T}^{0})/v_{T}\}dx
\end{equation*}%
is finite.

When $|x| >M$, $\mid Z_{T}+x\mid \leq h_{T}$ implies that $\mid Z_{T}\mid >|x| /2$ since $%
h_{T}=O(1)$. Hence, using assumption \ref{[A1"]}, 
\begin{equation*}
|x| P_{x}^{\ast }(\mid Z_{T}+x\mid \leq h_{T})\leq |x| P_{x}^{\ast
}(\mid Z_{T}\mid >|x| /2)\leq c_{0}\frac{|x| }{|x| ^{\kappa }},
\end{equation*}%
which in turns implies that 
\begin{equation*}
\int_{M\leq |x| \leq \delta v_{T}}P^{\ast }(\mid Z_{T}+x\mid \leq
h_{T})\pi\{b_{0}+(x+Z_{T}^{0})/v_{T}\} dx\leq C\int_{M\leq |x| \leq \delta
v_{T}}\frac{1}{|x| ^{\kappa -1}}dx\leq M^{-\kappa +2}.
\end{equation*}%
The same computation can be conducted in case \eqref{bis}.

\medskip

\noindent{Area (iii):} When $|x| \geq \delta v_{T}$ the second term is again
negligible for $\delta v_{T}$ large. Our focus then becomes 
\begin{equation*}
N_{3}={\frac{1}{h_{T}}}\int_{|x| \geq \delta v_{T}}{|x| }P_{x}^{\ast
}(\mid Z_{T}+x\mid \leq h_{T})\pi\{b_{0}+(x+Z_{T}^{0})/v_{T}\}dx.
\end{equation*}%
By assumption \ref{[A1']}, for some $\kappa >2$ we can bound $N_{3}$ as follows: 
\begin{flalign*}
N_{3}=&\frac{1}{h_{T}}\int_{|x| \geq \delta v_{T}}|x| P_{x}(\mid x+Z_{T}\mid \leq
h_{T})\pi\{b_{0}+(x+Z_{T}^{0})/v_{T}\}dx\\&\leq \frac{1}{h_{T}}\int_{|x| \geq
\delta v_{T}}\frac{ |x| c(b_{0}+(x+Z^{0}_{T})/v_{T})
}{(1+|x| -h_{T})^{\kappa}}\pi\{b_{0}+(x+Z_{T}^{0})/v_{T}\}dx\\&\lesssim\frac{v_{T}^{2}}{h_{T}}\int_{\mid b-\eta({y})\mid \geq\delta}\frac{{c(b)}
\mid b-\eta({y})\mid  }{\{1+v_{T}\mid b-\eta({y})\mid -h_{T}\}^{\kappa}}\pi(b)db.
\end{flalign*}Since $\eta ({y})=b_{0}+O_{P}(1/v_{T})$ we have, for $T$
large,

\begin{flalign*}
N_{3}&\lesssim\frac{v_{T}^{2}}{h_{T}}\int_{\mid b-b_{0}\mid \geq\delta/2}\frac{{c(b)}|b| \pi(b)}{(1+v_{T}\delta-h_{T})^{\kappa}}db\lesssim  \frac{v_{T}^{2}}{h_{T}}\left\{\int {c(b)} |b| \pi(b)db\right\}O(v_{T}^{-k})\lesssim O(v_{T}^{1-\kappa}\varepsilon_{T})=o(1),
\end{flalign*}
where assumptions \ref{[A7]} and \ref{[A1"]} ensure $\int {c(b)} |b| \pi(b)db<\infty$.
The same computation can be conducted in case \eqref{bis}.

Combining the results for the three areas we can conclude that $%
N_{T}/D_{T}=o_{P}(1)$ and the result follows.

\subsection{Proof of Theorem 4}
\label{sec:prTh4}

The proof follows the same lines as the proof of Theorem 3, with some extra
technicalities due to the multivariate nature of $\theta$. Define $%
G_0=\nabla_{\theta}b(\theta_0)$, $b_0=b(\theta_0)$ and let $Z_{T}^{0}=v_{T}\{\eta(y)-b_0\}$ and $$x(\theta) = v_T(\theta-\theta_0) -
(G_0^\intercal G_0)^{-1}G_0^\intercal Z_T^0.$$
We show that $E_{\Pi_{\epsilon}}\left\{ x(\theta) \right\} = o_p(1)$. We
write 
\begin{equation*}
E_{\Pi_{\epsilon}}\left\{ x(\theta) \right\} = \frac{ \int_{\Theta}
x(\theta) P_\theta\left\{ \| {\eta}( z)- {\eta(y)} \|\leq \varepsilon_T
\right\} \pi(\theta)d\theta }{ \int_{\Theta} P_\theta\left\{\| {\eta}( z)- {%
\eta(y)} \|\leq \varepsilon_T \right\} \pi(\theta)d\theta } = \frac{ N_T}{ D_T%
},
\end{equation*}
and study the numerator and denominator separately. Since for all $\epsilon
>0$ there exists $M_\epsilon>0$ such that, for all $M> M_\epsilon$, $%
P_{\theta_0}( \| Z_T^0\|>M/2) < \epsilon$, we can restrict ourselves to the
event $\|Z_T^0\|\leq M/2$ for some $M$ large.

We first study the numerator $N_T$ and we split $\Theta$ into $\{ \|G_0
x(\theta) \|\leq v_T \varepsilon_T - M\}$, $\{v_T\varepsilon_T - M \leq
\|G_0 x(\theta) \|\leq v_T \varepsilon_T + M\}$ and $\{ \|G_0 x(\theta) \|>
v_T \varepsilon_T + M\}$. The first integral is equal to 
\begin{equation*}
\begin{split}
I_1 &=\pi(\theta_0) \int_{ \|G_0 x(\theta) \|\leq v_T \varepsilon_T - M}
\{x(\theta) + O(v_T\varepsilon_T^2)\} P_\theta\left\{\| {\eta}( z)- {\eta(y)}%
\|\leq \varepsilon_T \right\} d\theta \\
& = \pi(\theta_0) \int_{ \|G_0 x(\theta) \|\leq v_T \varepsilon_T - M}
\{x(\theta) + O(v_T\varepsilon_T^2)\} d\theta \\
&-\pi(\theta_0) \int_{ \|G_0 x(\theta) \|\leq v_T \varepsilon_T - M}
\{x(\theta) + O(v_T\varepsilon_T^2)\} P_\theta\left\{ \| {\eta}( z)- {\eta(y)}
\|> \varepsilon_T \right\} d\theta .
\end{split}%
\end{equation*}
The first term in $I_{1}$ can be made arbitrarily small for $M$ large
enough. For the second term in $I_{1}$, we note 
\begin{equation*}
\begin{split}
v_{T}\varepsilon_{T}<\| v_T\{{\eta}( z)- {\eta(y)}\} \| & =\| Z_T - Z_T^0 +
v_TG_0 ( \theta - \theta_0) \| + O(\|\theta- \theta_0\|^2 ) \\
&= \| Z_T - P^{\perp}_{G_0} Z_T^0 + G_0x(\theta) \| + O(\|\theta-
\theta_0\|^2 ) \\
& \leq \|Z_T\| + \|P^{\perp}_{G_0} Z_T^0\| + \|G_0x(\theta) \| + O(\|\theta-
\theta_0\|^2 ) \\
& \leq \|Z_T\| + M/2 + \|G_0x(\theta) \| + O(\|\theta- \theta_0\|^2 ) \,,
\end{split}%
\end{equation*}
where $P^{\perp}_{G_0}$ is the orthogonal projection onto the space that is
orthogonal to $G_0$. Therefore, if $%
\|G_0x(\theta) \|\leq v_T \varepsilon_T - M$, then 
\begin{equation*}
M/2\leq v_T \varepsilon_T - M/2 - \|G_{0}x(\theta) \|\leq\|Z_T \|.
\end{equation*}
Hence, the second term of the right hand side of $I_1$ is bounded by a
term proportional to
\begin{equation*}
\begin{split}
& \int_{ \|G_{0} x(\theta) \|\leq v_T \varepsilon_T - M} 2\|G_{0} x(\theta)
\| P_\theta\left\{ \|Z_T\| > \varepsilon_Tv_T - M/2 - \|G_{0}x(\theta) \|
\right\} d\theta \\
& \lesssim \int_{ \|G_{0} x(\theta) \|\leq v_T \varepsilon_T - M} \frac{
\|G_{0} x(\theta) \|}{\{ v_T\varepsilon_T - M/2 - \|G_{0}x(\theta) \|
\}^\kappa } d\theta \\
&\lesssim v_T^{-k_\theta} \int_0^{v_T \varepsilon_T - M} \frac{ r^{k_\theta}%
}{ ( v_T\varepsilon_T - M/2 -r)^\kappa}dr \lesssim \varepsilon_T^{k_\theta}
M^{-\kappa}\,.
\end{split}%
\end{equation*}
The integral over $\{ \|G_{0} x(\theta) \|> v_T \varepsilon_T + M\}$, $I_3$,
is treated similarly. This leads to $\|I_1+I_3\| \leq M^{-\kappa}
\varepsilon_T^{k_\theta}$.

Likewise, using similar arguments we can show 
\begin{equation*}
D_T \gtrsim \int_{ \|G_{0} x(\theta) \|\leq v_T \varepsilon_T -
M}P_\theta\left\{\| \eta(z)-\eta(y) \|\leq \varepsilon_T \right\}
d\theta\gtrsim \varepsilon_T^{k_\theta}.
\end{equation*}

All that remains is to prove that the second integral $I_2$, the integral
over $\{v_{T}\varepsilon_{T}-M\leq \|G_0x(\theta)\|\leq
v_{T}\varepsilon_{T}+M\}$, is $o_p(\varepsilon_T^{k_\theta})$, with 
\begin{equation*}
\begin{split}
I_2 &= \int_{v_T \varepsilon_T - M\leq \|G_{0} x(\theta) \|\leq v_T
\varepsilon_T + M}\{x(\theta)+O(v_T\varepsilon_T^2)\} P_\theta\left\{\|
\eta(z)-\eta(y) \|\leq \varepsilon_T \right\} d\theta.
\end{split}%
\end{equation*}
Since 
\begin{equation*}
v_T^2\| \eta(z)-\eta(y) \|^2 = \| Z_T - P^{\perp}_{G_0}Z_T - G_{0} x(\theta)
\|^2 = \| Z_T- P^{\perp}_{G_0}Z_T^0\|^2 +\| G_{0} x(\theta) \|^2- 2 \langle
Z_T,G_{0} x(\theta) \rangle,
\end{equation*}%
where $\langle\cdot,\cdot\rangle$ is the inner product, setting $u =
(G_{0}^\intercal G_{0})^{1/2} x(\theta) \|G_{0} x(\theta)\|^{-1}$, $r =
\|G_{0} x(\theta)\|$, $\Gamma_0 = (G_{0}^\intercal
G_{0})^{-1/2}G_{0}^\intercal$, then, for $\mathcal{S }= \{u \in \mathbb{R}%
^{k_\theta} : \|u\|=1\}$, noting that $\theta = \theta(u,r)$ 
\begin{equation*}
\begin{split}
I_2 &= v_T^{-k_\theta}(G_{0}^\intercal G_{0})^{-1/2} \int_{v_T \varepsilon_T
- M}^{v_T \varepsilon_T + M}r^{k_\theta} \int_{u \in \mathcal{S}} u
P_\theta\left(\| Z_T- P^{\perp}_{G_0}Z_T^0\|^2 +r^2- 2r \langle\Gamma_0Z_T,u
\rangle \leq v_T^2\varepsilon_T^2 \right) dudr\\
& \quad + O\left( v_T\varepsilon_T^{2+k_\theta}\right) \\
& = v_T^{-k_\theta}(G_{0}^\intercal G_{0})^{-1/2} \int_{ - M}^{ M}(v_T
\varepsilon_T+ r)^{k_\theta}\times \\
& \quad \int_{u \in \mathcal{S}} u P_\theta\left(\| Z_T-
P^{\perp}_{G_0}Z_T^0\|^2 - 2r \langle\Gamma_0Z_T,u \rangle - 2 \varepsilon_T
v_T \langle\Gamma_0Z_T,u \rangle \leq - r^2 - 2 rv_T\varepsilon_T \right) dudr \\
& \quad + O\left( v_T\varepsilon_T^{2+k_\theta}\right),
\end{split}%
\end{equation*}
where $du$ denotes the Lebesgue measure on $\mathcal{S}$.
Moreover, we have 
\begin{equation*}
\begin{split}
& P_\theta\left(\| Z_T- P^{\perp}_{G_0}Z_T^0\|^2 - 2r\langle\Gamma_0Z_T,u
\rangle - 2 \varepsilon_T v_T \langle\Gamma_0Z_T,u \rangle \leq - r^2 - 2
rv_T\varepsilon_T \right) \\
& = P_\theta\left\{ \langle\Gamma_0Z_T,u \rangle \geq \frac{r\varepsilon_T
v_T}{r+\varepsilon_T v_T} + \frac{\| Z_T- P^{\perp}_{G_0}Z_T^0\|^2 + r^2 }{
2(v_T\varepsilon_T+r) } \right\}
\end{split}%
\end{equation*}
and for any $a_T>M$ with $a_T = o(v_T \varepsilon_T)$, 
\begin{equation*}
\begin{split}
& P_{\theta}\left(\| Z_T- P^{\perp}_{G_0}Z_T^0\|^2 \geq a_T \right) \lesssim
c_0 a_T^{-\kappa/2}, \\
& \left| P_\theta\left( \langle\Gamma_0Z_T,u \rangle \geq r \right) -
P_\theta\left\{ \langle\Gamma_0Z_T,u \rangle \geq r - 2a_T/(v_T
\varepsilon_T) \right\} \right| = o(1),
\end{split}%
\end{equation*}
with for all $r$ and $u $, $P_\theta\left( \langle\Gamma_0 Z_T,u \rangle
\geq r \right) = \{1 -\Phi( r/ \| \Gamma_0 A(\theta_0)^{1/2}\| )\}+ o(1),$ uniformly over $\|\theta - \theta_0\|\leq \delta$
and $A(\theta_0)$ as in the proof of Theorem \ref{normal_thm}. Since for all 
$r \in [-M,M]$, $(v_T \varepsilon_T+ r)^{k_\theta} = (v_T
\varepsilon_T)^{k_\theta} + O(M (v_T \varepsilon_T)^{k_\theta-1}),$ the
dominated convergence theorem implies 
\begin{equation*}
\begin{split}
I_2 & = \varepsilon_T^{k_\theta}(G_{0}^\intercal G_{0})^{-1/2} \int_{ - M}^{
M} \int_{u \in \mathcal{S}} u \left[1 -\Phi\{ r/ \| \Gamma_0 A(\theta_0)^{1/2}\|
\}\right]dudr + o(\varepsilon_T^{k_\theta}) =
o(\varepsilon_T^{k_\theta}),
\end{split}%
\end{equation*}
which completes the proof.

\subsection{Proof of Corollary 1}

Consider first the case where $\varepsilon _{T}=o(v_{T}^{-1})$. Using the
same types of computations as in the proof of Theorem \ref{normal_thm"}, case (v), in this
Supplementary Material, we have, for $Z_{T}={\Sigma }_{T}({\theta }_{0})\{{%
\eta }({z})-{b}({\theta })\},$ 
\begin{equation*}
\begin{split}
\alpha _{T}& =\int_{{\Theta }}P_{{\theta }}\left[ \Vert
Z_{T}-Z_{T}^{0}-v_{T}\{{b}({\theta })-{b}({\theta }_{0})\}\Vert \leq
\varepsilon _{T}v_{T}\right] \pi ({\theta })d{\theta } \\
& \asymp (\varepsilon _{T}v_{T})^{k_{\eta }}\int_{{\Theta }}\varphi
\{Z_{T}^{0}+v_{T}\nabla _{\theta }{b}({\theta }_{0})({\theta }-{\theta }%
_{0})\}d{\theta _{{}}}\asymp \varepsilon _{T}^{k_{\eta }}v_{T}^{k_{\eta }-k_{%
{\theta }}}\lesssim v_{T}^{-k_{\theta }}.
\end{split}%
\end{equation*}%
In the case where $\varepsilon _{T}\gtrsim v_{T}^{-1}$, then the computations in cases (ii) and (iii) in the proof of
Theorem \ref{normal_thm"} imply
\begin{equation*}
\alpha _{T}=P_{\theta }\left\{ \Vert {\eta }({z})-{\eta }({y})\Vert \leq
\varepsilon _{T}\right\} \asymp \int_{\Theta }\varphi \{Z_{T}^{0}+
\nabla_{\theta}{b}({\theta }_{0})v_{T}({\theta }-{\theta }_{0})\}d{\theta }\asymp
\varepsilon _{T}^{k_{{\theta }}}.
\end{equation*}


\section{Illustrative Example}

In this section we illustrate the implications of Theorems 1--3 in the
moving model of order two that was introduced in Example 1. Consider
observations from the data generating process 
\begin{equation}
y_{t}=e_{t}+\theta _{1}e_{t-1}+\theta _{2}e_{t-2}\;\;(t=1,\dots,T) ,
\label{MA2"}
\end{equation}%
where $e_{t}\sim \mathcal{N}(0,1)$ is independently and identically
distributed. Our prior belief for $\theta =(\theta _{1},\theta
_{2})^{\intercal }$ is uniform over the invertibility region 
\begin{equation}
\left\{(\theta_1,\theta_2)^\intercal:\ -2\leq \theta _{1}\leq 2,\;\theta
_{1}+\theta _{2}\geq -1,\theta _{1}-\theta _{2}\leq 1\right\}.
\label{const1"}
\end{equation}%
We follow Marin \textit{et al.} (2011) and choose as summary statistics for
Algorithm 1 the sample autocovariances $\eta _{j}({y})=\frac{1}{T}%
\sum_{t=1+j}^{T}y_{t}y_{t-j}$, for $j=0,1,2$, so that ${\eta }( {y}) =\{\eta
_{0}({y}),\eta _{1}({y}),\eta _{2}({y})\}^{\intercal }$. The binding
function $b(\theta )$ then has the simple analytical form: 
\begin{equation*}
\theta \mapsto b(\theta )=%
\begin{bmatrix}
E_{\theta }(z_{t}^2) \\ 
E_{\theta }(z_{t}z_{t-1}) \\ 
E_{\theta }(z_{t}z_{t-2})%
\end{bmatrix}%
\equiv 
\begin{pmatrix}
1+\theta _{1}^{2}+\theta _{2}^{2} \\ 
\theta _{1}+\theta _{1}\theta _{2} \\ 
\theta _{2}%
\end{pmatrix}%
.
\end{equation*}

The following subsections demonstrate the implications of the limit results
in the main text within the confines of the above example. By simultaneously
shifting the sample size $T$ and the tolerance parameter $\varepsilon _{T}$
we can graphically illustrate Theorems 1--3.

Each demonstration considers minor variants of the following general
simulation design: the true parameter vector generating the observed data is
fixed at $\theta _{0}=(\theta _{1,0},\theta _{2,0})^{\intercal
}=(0.6,0.2)^{\intercal }$; for a given sample size, $T$, of 500, 1000,
50000, observed data, ${y}=(y_{1},\dots ,y_{T})^{\intercal }$, is generated
from the process in equation \eqref{MA2"}; the posterior density is estimated via Algorithm 1
with the tolerance chosen to be a particular order of $T$, and using $%
N=50,000$ Monte Carlo draws taken from uniform priors satisfying %
\eqref{const1"}. In these examples we take $d_{2}\{\eta (z),\eta (y)\}=\Vert
\eta (z)-\eta (y)\Vert $.

A central result in the main text is that the choice of $\varepsilon _{T}$ 
drives the large sample behavior of the approximate
posterior distribution and its mean. To highlight this fact, our numerical
experiments will use different choices for the tolerance. In particular, and
with reference to the illustration of Theorem 2, in the main text, the choices of $\varepsilon
_{T}$ are $\{1/T^{0.4},1/T^{0.5},1/T^{0.55}\}$. In this example, we
have that $v_{T}=T^{0.5}$ and the three tolerance choices represent 
{respectively } cases {(i)}, {(ii)} and {(iii)} of Theorem
2. Our use of {different} tolerances relays
the distinction between what condition on $%
\varepsilon _{T}$ is required for posterior concentration, and what is
required to yield asymptotic normality of the posterior measure. With regard
to Theorem 3,\ the choice $\varepsilon _{T}=1/T^{0.4}$ 
highlights that asymptotic normality of the posterior mean can be achieved
despite a lack of Gaussianity for the posterior measure itself.

\subsection{Theorem 1}

Theorem 1 implies that under regularity, as $T\rightarrow \infty $, the
posterior measure $\Pi _{\varepsilon
}\{\cdot \mid \eta(y)\}$ concentrates on sets
containing $\theta _{0}$, namely $\Pi _{\varepsilon }\{d_{1}(\theta
,\theta)\leq\delta\mid\eta(y)\}=1+o_{P}(1)$ for all $\delta >0$,
provided $\varepsilon _{T}=o(1)$. To demonstrate this
concentration result, we take $\varepsilon _{T}=1/T^{0.4}$ and run Algorithm
1, taking $N=50,000$ draws from the prior. The results are presented in Fig. %
\ref{fig1"}. To keep the Monte Carlo error at a constant level, for each
sample size we retain 100 simulated values of $\theta $\ that lead to
realizations of $\Vert \eta (y)-\eta (z)\Vert $\ below the tolerance, in
agreement with the nearest-neighbor interpretation of algorithm 1.%

\begin{figure}[tbp]
\centering 
\setlength\figureheight{4.0cm} 
\setlength\figurewidth{6.5cm} 
\input{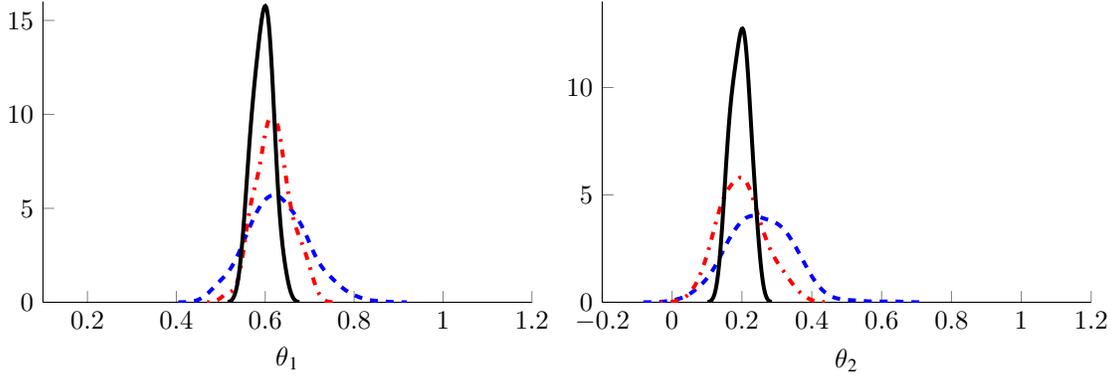}
\caption{Posterior concentration demonstration. Estimated approximate
posterior distributions across sample sizes $T$=500 ({\color{blue}\textbf{-
- -}}); $T$=1000 ({\color{red}\textbf{\--- $\cdot$ \---}}); $T$=5000 ({\ 
\color{black}\textbf{\----}}).}
\label{fig1"}
\end{figure}

Figure \ref{fig1"} shows that the posterior measure $\Pi _{\varepsilon
}\{\cdot \mid \eta(y)\}$ is concentrating on $\theta
_{0}=(0.6,0.2)^{\intercal }$ as $T$ increases. The results in 
Fig. \ref{fig1"}
reflect the fact that a tolerance proportional to $\varepsilon _{T}=1/T^{0.4}
$ will be small enough to yield posterior concentration. However, and with
reference to Theorem 2, this tolerance may or may not yield {%
asymptotic normality and, hence, correct asymptotic coverage of credible
intervals.} We explore this issue in the following section.


\subsection{Theorem 2}

Theorem 2 states that the shape of the approximate posterior measure is
determined in large part by the speed at which $\varepsilon _{T}$ goes to $0$. 
If this convergence is too slow, then the posterior measure will have {a
non-Gaussian} limiting shape. 

This result can be visualized by considering two alternative values for the
tolerance: $\varepsilon _{T}=1/T^{0.4}$ {and} $\varepsilon
_{T}=1/T^{0.55}$. Figure \ref{fig2} and Fig. \ref{fig3} display the resulting
approximate posterior density estimates using these two tolerance rules for
sample size $T=500$ and $T=1000$, respectively. 
\begin{figure}[h]
\centering 
\setlength\figureheight{3.2cm} 
\setlength\figurewidth{5.6cm} 
\input{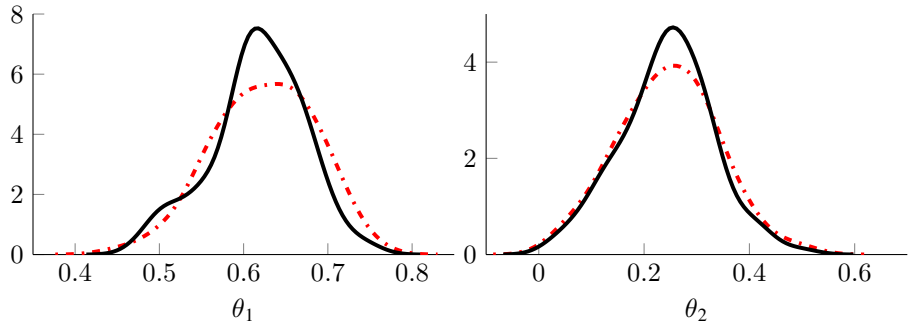}
\caption{Comparison of two tolerance rules for $\protect\varepsilon %
_{T}$ : $\protect\varepsilon _{T}=1/T^{0.4}$ ({\color{red}\textbf{\--- $%
\cdot $ \---}}); $\protect\varepsilon _{T}=1/T^{0.55}$ ({\color{black}%
\textbf{\----}}); The sample size is $T=500$. }
\label{fig2}
\end{figure}
\begin{figure}[h]
\centering 
\setlength\figureheight{3.2cm} 
\setlength\figurewidth{5.6cm} 
\input{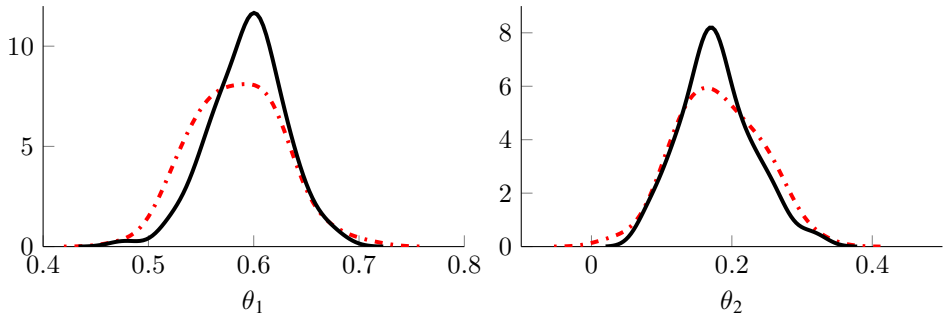}
\caption{Same information as Fig. \protect\ref{fig2} but for $T=1000$.}
\label{fig3}
\end{figure}
Figure \ref{fig2} demonstrates that at $T=500$ neither posterior density,
for $\theta _{1}$ or $\theta _{2}$ and across both tolerance rules, has a
shape that is particularly Gaussian. However, at $T=1000$, and for both $%
\theta _{1}$ and $\theta _{2}$, Fig. \ref{fig3} demonstrates that the
posterior densities based on the tolerance $\varepsilon _{T}=1/T^{0.55}$,
which satisfies the conditions for the Bernstein--von Mises result, appear
to be approximately Gaussian. In contrast, the approximate posterior
densities constructed from $\varepsilon _{T}=1/T^{0.4}$ display non-Gaussian
features.%

From a practical perspective, the key result of Theorem 2 is that for
credible regions built from $\Pi _{\varepsilon }\{\cdot \mid \eta(y)\}$ to
have asymptotically correct frequentist coverage, it must be that $%
\varepsilon _{T}=o(1/v_{T})$, where $v_{T}$ is such that $\Vert \eta
(z)-b(\theta )\Vert =O_{P}(1/v_{T})$. In the moving average model example, $%
v_{T}=T^{0.5}$ and Theorem 2 implies that choosing a tolerance $%
\varepsilon_{T}=1/T^{0.4}$, which corresponds to case (i) of Theorem 2, will
yield credible sets whose coverage converges to one asymptotically; choosing
a tolerance of $\varepsilon_{T}=1/T^{.55}$, which corresponds to case (iii)
of Theorem 2, will lead to asymptotically correct coverage rates; a
tolerance of $\varepsilon_{T}=1/T^{0.5}$ will yield coverage that is
asymptotically of the correct magnitude, in that the coverage will not be
zero or one, but will in general differ from the nominal level.

To demonstrate this point we generate 1000 observed artificial data sets
with sample sizes $T=500$ and $T=1000$, and for each data set we run
Algorithm 1 for {all} three alternative values of $\varepsilon _{T}$.
For a given sample, and a given tolerance, we produce the approximate
Bayesian computation posterior density in the manner described above and
compute the 95\% credible intervals for $\theta _{1}$ and $\theta _{2}$. The
average length and the Monte Carlo coverage rate, across the 1000
replications, is then recorded in Table \ref{tab2} for each scenario. The
average length of the credible regions is clearly larger, and the Monte
Carlo coverage further from the nominal value of 95\%, the further is the
tolerance from the value required to produce asymptotic Gaussianity, namely $%
\varepsilon _{T}=1/T^{0.55}$, which provides numerical support for the
theoretical results.

\begin{table}[tbph]
\tbl{Gaussianity of the approximate posterior distributions: the tolerances are $\protect\varepsilon _{1}=1/T^{0.4}$, $\protect\varepsilon_{2}=1/T^{0.5}$ and $\protect\varepsilon _{3}=1/T^{0.55}$}{\begin{tabular}{rrrrrrr}
& Width &  & &Cov. & & \\ 
T=500 & $\varepsilon_1$ & $\varepsilon_2$ & $\varepsilon_3$ & $\varepsilon_1$ & $\varepsilon_2$ & $\varepsilon_3$ \\ 
$\theta_1$ & 0.2602 & 0.2294 & 0.2198 & 96.30 & 95.60 & 95.60 \\ 
$\theta_2$ & 0.3212 & 0.3108 & 0.3086 & 98.30 & 97.00 & 96.00 \\ 
T=1000 &  &  &  &  &  &  \\ 
$\theta_1$ & 0.1823 & 0.1573 & 0.1484 & 96.80 & 96.20 & 95.50 \\ 
$\theta_2$ & 0.2366 & 0.2244 & 0.2219 & 96.60 & 94.30 & 94.50 \\&&&&&& \\
\end{tabular}}%
\begin{tabnote}
Width stands for average length and Cov.~for Monte Carlo coverage rate.
\end{tabnote}
\label{tab2}
\end{table}

\subsection{Theorem 3}

The key result of Theorem 3 is that even when $\Pi _{\varepsilon }\{\cdot \mid
\eta(y)\}$ is not asymptotically Gaussian, the posterior mean
associated with Algorithm 1, $\hat{\theta}=E_{\Pi_{\varepsilon} }(\theta),$
can still be asymptotically Gaussian, and asymptotically unbiased so long as 
$\lim_{{T}}v_{T}\varepsilon _{T}^{2}=0$. However, as proven in Theorem 2,
the corresponding confidence regions and uncertainty measures built from $%
\Pi _{\varepsilon }\{\cdot \mid\eta(y)\}$ will only be an adequate
reflection on the actual uncertainty associated with $\hat{\theta}$ if $%
\varepsilon _{T}=o(1/v_{T})$.

In this section we once again generate 1000 observed data sets of a given
sample size ($T=500$ and $T=1000$) according to equation \eqref{MA2"} and $\theta _{0}=(0.6,0.2)^{\intercal }$,
and produce 1000 posterior densities based on the tolerance $\varepsilon _{T}
$ being one of $\{1/T^{0.4},1/T^{0.5},1/T^{0.55}\}$. For each of the three
values of $\varepsilon _{T}$, and for a sample size of $T=500$, we record
the posterior mean across the 1000 replications and plot the relevant
empirical densities in Fig. \ref{fig4}. Figure \ref{fig5} contains the results for $T=1000$%
.

Figure \ref{fig4} demonstrates that the standardized Monte Carlo sampling
distribution of $\hat{\theta}=E_{\Pi_{\varepsilon} }(\theta)$, over the 1000
replications, and for each of the three values of $\varepsilon _{T} $, is
approximately Gaussian for both parameters and centered at zero. This
accords with the theoretical results, which only require that $\lim_{{T}%
}\varepsilon_{T}=0$, for asymptotic Gaussianity, and $\lim_{{T}%
}v_{T}\varepsilon _{T}^{2}=0$, for zero asymptotic bias, a condition that is
satisfied for each of the three tolerance values. This result is also in
evidence for $T=1000$, as can be seen in Fig. \ref{fig5}.

\begin{figure}[h!]
\centering 
\setlength\figureheight{3.2cm} 
\setlength\figurewidth{5.6cm} 
\input{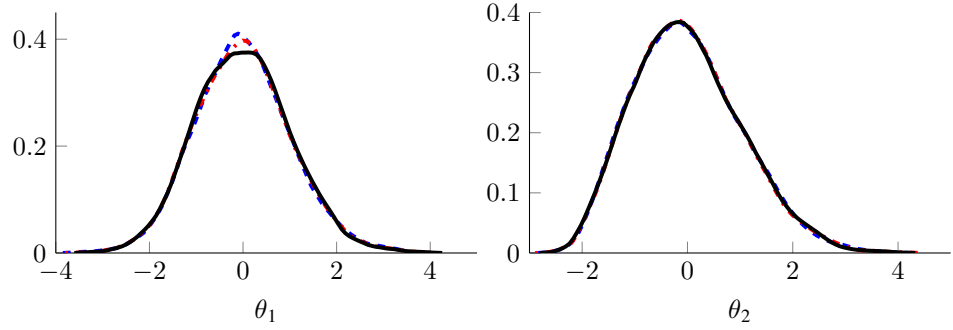}
\caption{Comparison of different tolerance rules for $\protect\varepsilon%
_{T} $: $\protect\varepsilon_{T}=1/T^{0.4}$ ({\color{blue}\textbf{- - -}}); $%
\protect\varepsilon_{T}=1/T^{0.5}$ ({\color{red}\textbf{\--- $\cdot$ \---}}%
); $\protect\varepsilon_{T}=1/T^{0.55}$ ({\color{black}\textbf{\----}}); The
sample size is $T=500$.}
\label{fig4}
\end{figure}

\begin{figure}[h!]
\centering 
\setlength\figureheight{3.2cm} 
\setlength\figurewidth{5.6cm} 
\input{ABC_mean1_1000_BKA.tikz}
\caption{Same information as Fig. \protect\ref{fig4} but for $T=1000$.}
\label{fig5}
\end{figure}

\end{document}